\documentclass[11pt,letterpaper,onecolumn]{quantumarticle}
\pdfoutput=1
\usepackage{fullpage}
\usepackage{enumitem}
\usepackage[utf8]{inputenc}
\usepackage{amssymb}
\usepackage{amsmath}
\usepackage[margin=0.7in]{geometry}
\usepackage{color}
\usepackage{commath}
\usepackage{graphicx}
\usepackage{float}
\usepackage{amsthm}
\usepackage{comment}
\usepackage{braket}
\usepackage{hyphenat}
\usepackage{ragged2e}
\usepackage{doi}

\newcommand{\Id}{\ensuremath{\mathop{\rm Id}\nolimits}}

\newtheorem{theorem}{Theorem}[section]
\newtheorem{lemma}[theorem]{Lemma}
\newtheorem{claim}[theorem]{Claim}

\theoremstyle{definition}
\newtheorem{definition}[theorem]{Definition}
\newtheorem{protocol}[theorem]{Protocol}
\newtheorem{ass}[theorem]{Assumptions}

\theoremstyle{remark}
\newtheorem{remark}[theorem]{Remark}

\newcommand{\tnote}[1]{\textcolor{magenta}{\small {\textbf{(Thomas:} #1\textbf{) }}}}

\newcommand{\znote}[1]{\textcolor{blue}{\small {\textbf{(Tina:} #1\textbf{) }}}}

\title{Classical zero-knowledge arguments for quantum computations}
\author{Thomas Vidick}
\affiliation{Department of Computing and Mathematical Sciences, California Institute of Technology, USA}
\email{vidick@caltech.edu}

\author{Tina Zhang}
\affiliation{Division of Physics, Mathematics and Astronomy, California Institute of Technology, USA}
\email{tinazhang@caltech.edu}

\date{}

\begin{document}
\maketitle
\begin{abstract}
We show that every language in QMA admits a classical-verifier, quantum-prover zero-knowledge argument system which is sound against quantum polynomial-time provers and zero-knowledge for classical (and quantum) polynomial-time verifiers. The protocol builds upon two recent results: a computational zero-knowledge proof system for languages in QMA, with a quantum verifier, introduced by Broadbent et al.\ (FOCS 2016), and an argument system for languages in QMA, with a classical verifier, introduced by Mahadev (FOCS 2018).
\end{abstract}
\section{Introduction}

The paradigm of the interactive proof system is a versatile tool in complexity theory. Although traditional complexity classes are usually defined in terms of a single Turing machine—NP, for example, can be defined as the class of languages which a non-deterministic Turing machine is able to decide—many have reformulations in the language of interactive proofs, and such reformulations often inspire natural and fruitful variants on the traditional classes upon which they are based. (The class MA, for example, can be considered a natural extension of NP under the interactive-proof paradigm.)

Intuitively speaking, an interactive proof system is a model of computation involving two entities, a \emph{verifier} and a \emph{prover}, the former of whom is computationally efficient, and the latter of whom is unbounded and untrusted. The verifier and the prover exchange messages, and the prover attempts to `convince' the verifier that a certain problem instance is a yes-instance. We can define some particular complexity class as the set of languages for which there exists an interactive proof system that 1) is \emph{complete}, 2) is \emph{sound}, and 3) has certain other properties which vary depending on the class in question. Completeness means, in this case, that for any problem instance in the language, there is an interactive proof involving $r$ messages in total that the prover can offer the verifier which will cause it to accept with at least some probability $p$; and soundness means that, for any problem instance not in the language, no prover can cause the verifier to accept, except with some small probability $q$. For instance, if we require that the verifier is a deterministic polynomial-time Turing machine, and set $r = 1, p = 1$, and $q = 0$, the class that we obtain is of course the class NP. If we allow the verifier to be a probabilistic polynomial-time machine, and set $r = 1, p = \frac{2}{3}$, $q = \frac{1}{3}$, we have MA. Furthermore, if we allow the verifier to be an efficient \emph{quantum} machine, and we allow the prover to communicate with it quantumly, but we retain the parameter settings from MA, we obtain the class QMA. Finally, if we allow $r$ to be any polynomial in $n$, where $n$ is the size of the problem instance, but otherwise preserve the parameter settings from MA, we obtain the class IP.

For every complexity class thus defined, there are two natural subclasses which consist of the languages that admit, respectively, a \emph{statistical} and a \emph{computational} \emph{zero-knowledge} interactive proof system with otherwise the same properties. The notion of a zero-knowledge proof system was first considered by Goldwasser, Micali and Rackoff in~\cite{goldwasser1989knowledge}, and formalises the surprising but powerful idea that the prover may be able to prove statements to the verifier in such a way that the verifier learns nothing except that the statements are true. Informally, an interactive proof system is \emph{statistical zero-knowledge} if an arbitrary malicious verifier is able to learn from an honest prover that a problem instance is a yes-instance, but can extract only negligible amounts of information from it otherwise; and the computational variant provides the same guarantee only for malicious polynomial-time verifiers. For IP in particular, the subclass of languages which admit a statistical zero-knowledge proof system that otherwise shares the same properties had by proof systems for languages in IP is known as SZK. Its computational sibling, meanwhile, is known as CZK. It is well-known that, contingent upon the existence of one-way functions, NP $\subseteq$ CZK: computational zero-knowledge proof systems have been known to exist for every language in NP since the early 1990s (\cite{original-zk}). However, because these proof systems often relied upon intractability assumptions or techniques (e.g. `rewinding') that failed in quantum settings, it was not obvious until recently how to obtain an analogous result for QMA. One design for a zero-knowledge proof system for promise problems in QMA was introduced by Broadbent, Ji, Song and Watrous in \cite{qma}. Their work establishes that, provided that a quantum computationally concealing, unconditionally binding commitment scheme exists, QMA $\subseteq$ QCZK.

There are, of course, a myriad more variations on the theme of interactive proofs in the quantum setting, each of which defines another complexity class. For example, motivated partly by practical applications, one might also consider the class of languages which can be decided by an interactive proof system involving a classical verifier and a quantum prover communicating classically, in which the soundness condition still holds against arbitrary provers, but the honest prover can be implemented in quantum polynomial time. (For simplicity, we denote this class by $\mathrm{IP_{BQP}}$.) The motivation for this specific set of criteria is as follows: large-scale quantum devices are no longer so distant a dream as they seemed only a decade ago. If and when we have such devices, how will we verify, using our current generation of classical devices, that our new quantum computers can indeed decide problems in BQP? This problem—namely, the problem of showing that BQP $\subseteq \mathrm{IP_{BQP}}$—is known informally as the problem of \emph{quantum verification}.

The problem of quantum verification has not yet seen a solution, but in recent years a number of strides have been made toward producing one. As of the time of writing, protocols are known for the following three variants on the problem:

\begin{enumerate}[itemsep=1pt, topsep=0pt]
\item It was shown in~\cite{AharonovBE10,aharonov2017interactive} that a classical verifier holding a quantum register consisting only of a constant number of qubits can decide languages in BQP by communicating quantumly with a single BQP prover. In~\cite{broadbent2009universal,fitzsimons2017unconditionally}, this result was extended to classical verifiers with \emph{single-qubit} quantum registers. All of these protocols are sound against arbitrary provers.
\item It was shown in~\cite{reichardt2013classical} that an entirely classical verifier can decide languages in BQP by interacting classically with two entangled, non-communicating QPT provers. This protocol is likewise sound against arbitrary provers.

\item It was shown in \cite{measurement} that an entirely classical verifier can decide languages in BQP by executing an \emph{argument system} (\cite{arguments}) with a single BQP prover. An argument system differs from a proof system in that 1) its honest prover must be efficient, and 2) an argument system need not be sound against arbitrary provers, but only efficient ones. In this case, the argument system in \cite{measurement} is sound against quantum polynomial-time provers. (The class of languages for which there exists an argument system involving a classical probabilistic polynomial-time verifier and a quantum polynomial-time prover is referred to throughout \cite{measurement} as $\mathrm{QPIP_0}$.) The argument system introduced in \cite{measurement} is reliant upon cryptographic assumptions about the quantum intractability of Learning With Errors (LWE; see~\cite{regev2009lattices}) for its soundness. For practical purposes, if this assumption holds true, the problem of verification can be considered solved.
\end{enumerate}

The last of these three results establishes that BQP $\subseteq \mathrm{QPIP_0}$, contingent upon the intractability of LWE. (As a matter of fact, the same result also establishes that QMA $\subseteq \mathrm{QPIP_0}$, provided the efficient quantum prover is given access to polynomially many copies of a quantum witness for the language to be verified, in the form of ground states of an associated local Hamiltonian.) In this work, we show that the protocol which \cite{measurement} introduces for this purpose can be combined with the zero-knowledge proof system for QMA presented in \cite{qma} in order to obtain a \emph{zero-knowledge} argument system for QMA. It follows naturally that, if the LWE assumption holds, and quantum computationally hiding, unconditionally binding commitment schemes exist, \footnote{It is known that quantum computationally hiding, unconditionally binding commitment schemes fitting our requirements can be constructed from LWE. See, for example, Section 2.4.2 in \cite{coladangelo2019non}.} QMA $\subseteq$ CZK-$\mathrm{QPIP_0}$, where the latter refers to the class of languages for which there exists a \emph{computational zero-knowledge} interactive argument system involving a classical verifier and a quantum polynomial-time prover. Zero-knowledge protocols for languages in NP are an essential component of many cryptographic constructions, such as identification schemes, and are often used in general protocol design (for example, one can force a party to follow a prescribed protocol by requiring it to produce a zero-knowledge proof that it did so). Our result opens the door for the use of zero-knowledge proofs in protocols involving classical and quantum parties which interact classically in order to decide languages defined in terms of quantum information (for instance, to verify that one of the parties possesses a quantum state having certain properties).

We now briefly describe our approach to the problem. The proof system for promise problems in QMA presented in \cite{qma} is \emph{almost} classical, in the sense that the only quantum action which the honest verifier performs is to measure a quantum state after applying Clifford gates to it. The key contribution which \cite{measurement} makes to the problem of verification is to introduce a \emph{measurement protocol} which, intuitively, allows a classical verifier to obtain honest measurements of its prover's quantum state. The combining of the proof system from \cite{qma} and the measurement protocol from \cite{measurement} is therefore a fairly natural action. 

That the proof system of \cite{qma} is complete for  problems in QMA follows from the QMA-completeness of a problem which the authors term the \emph{5-local Clifford Hamiltonian problem}. However, the argument system which \cite{measurement} presents relies upon the QMA-completeness of the well-known \emph{2-local XZ Hamiltonian problem} (see Definition \ref{def:2-local-xz-problem}). For this reason, the two results cannot be composed directly. Our first step is to make some modifications to the protocol introduced in \cite{qma} so it can be used to verify that an XZ Hamiltonian is satisfied, instead of verifying that a Clifford Hamiltonian is satisfied. We then introduce a composite protocol which replaces the quantum measurement in the protocol from \cite{qma} with an execution of the measurement protocol from \cite{measurement}. With the eventual object in mind of proving that the result is sound and zero-knowledge, we introduce a \emph{trapdoor check} step into our composite protocol, and split the \emph{coin-flipping protocol} used in the proof system from \cite{qma} into two stages. We explain these decisions briefly here, after we present a summary of our protocol and its properties, and refer the reader to Sections \ref{section:protocol}, \ref{section:soundness} and \ref{section:zk} for fuller expositions.

\begin{protocol}
\emph{Zero-knowledge, classical-verifier argument system for QMA (informal summary).}
\label{protocol:main-protocol-informal}

\emph{Parties.}

The protocol involves
\begin{enumerate}[itemsep=1pt, topsep=0pt]
\item A \emph{verifier}, which runs in classical probabilistic polynomial time;
\item A \emph{prover}, which runs in quantum polynomial time.
\end{enumerate}

\emph{Inputs.}
The protocol requires the following primitives: 
\begin{itemize}[itemsep=1pt, topsep=0pt]
\item A perfectly binding, quantum computationally concealing commitment protocol.
\item A zero-knowledge proof system for NP.
\item An extended trapdoor claw-free function family (ETCFF family), as defined in \cite{measurement}.
\end{itemize}
Apart from the above cryptographic primitives, we assume that the verifier and the prover also receive the following inputs.
\begin{enumerate}[noitemsep,topsep=0pt]
\item Input to the verifier: a 2-local XZ Hamiltonian $H$ (see Definition \ref{def:2-local-xz-problem}), along with two numbers, $a$ and $b$, which define a promise about the ground energy of $H$. Because the 2-local XZ Hamiltonian promise problem is complete for QMA, any input to any decision problem in QMA can be reduced to an instance of the 2-local XZ Hamiltonian problem.
\item Input to the prover: the Hamiltonian $H$, the numbers $a$ and $b$, and the quantum state $\rho = \sigma^{\otimes m}$, where $\sigma$ is a ground state of the Hamiltonian $H$. 
\end{enumerate}

\emph{Protocol.} 
\begin{enumerate}[itemsep=1pt, topsep=0pt]
\item The prover applies an encoding process to $\rho$. Informally, the encoding can be thought of as a combination of an encryption scheme and an authentication scheme: it both hides the witness state $\rho$ and ensures that the verifier cannot meaningfully tamper with the measurement results that it reports in step 5. Like most encryption and authentication schemes, this encoding scheme is keyed. For convenience, we refer to the encoding procedure determined by a particular encoding key $K$ as $E_K$.\footnote{The notation used here for the encoding key is not consistent with that which is used later on; it is simplified for the purposes of exposition.}
\item The prover commits to the encoding key $K$ from the previous step using a classical commitment protocol, and sends the resulting commitment string $z$ to the verifier.
\item The verifier and the prover jointly decide which random terms from the Hamiltonian $H$ the verifier will check by executing a coin-flipping protocol. (`Checking terms of $H$' means that the verifier obtains measurements of the state $E_K(\rho)$ and checks that the outcomes are distributed a particular way---or, alternatively, asks the prover to prove to it that they are.) However, because it is important that the prover does not know which terms will be checked before the verifier can check them, the two parties only execute the first half of the coin-flipping protocol at this stage. The verifier commits to its part of the random string, $r_v$, and sends the resulting commitment string to the prover; the prover sends the verifier $r_p$, its own part of the random string; and the verifier keeps the result of the protocol $r = r_v \oplus r_p$ secret for the time being. The random terms in the Hamiltonian which the verifier will check are determined by $r$.
\item The verifier and the prover execute the measurement protocol from \cite{measurement}. Informally, this allows the verifier to obtain honest measurements of the qubits of the prover's encoded witness state, so that it can check the Hamiltonian term determined by $r$. The soundness guarantee of the measurement protocol prevents the prover from cheating, even though the prover, rather than the verifier, is physically performing the measurements. This soundness guarantee relies on the security properties of a family of trapdoor one-way functions termed an ETCFF family in \cite{measurement}. Throughout the measurement protocol, the verifier holds trapdoors for these one-way functions, but the prover does not, and this asymmetry is what allows the (intrinsically weaker) verifier to ensure that the prover does not cheat.
\item The verifier opens its commitment to $r_v$, and also sends the prover its measurement outcomes $u$ and function trapdoors from the previous step.
\item The prover checks, firstly, that the verifier's trapdoors are valid, and that it did not tamper with the measurement outcomes $u$. (It can determine the latter by making use of the authentication-scheme-like properties of $E_K$ from step 1.) If both tests pass, it then proves the following statement to the verifier, using a zero-knowledge proof system for NP:
\begin{enumerate}[itemsep=1pt, topsep=0pt]
\item[] There exists a string $s_p$ and an encoding key $K$ such that $z = \textsf{commit}(K, s_p)$ and $Q(K, r, u) = 1.$
\end{enumerate}
The function $Q$ is a predicate which, intuitively, takes the value 1 if and only if both the verifier \emph{and} the prover were honest. In more specific (but still informal) terms, $Q(K, r, u)$ takes the value 1 if $u$ contains the outcomes of honest measurements of the state $E_K(\rho)$, where $\rho$ is a state that passes the set of Hamiltonian energy tests determined by $r$.
\end{enumerate}
\end{protocol}

\begin{lemma}[soundness; informal]
Assume that LWE is intractable for quantum computers. Then, in a no-instance execution of Protocol \ref{protocol:main-protocol-informal}, the probability that the verifier accepts is at most a function that is negligibly close to $\frac{3}{4}$.
\end{lemma}

\begin{lemma}[zero-knowledge; informal]
Assume that LWE is intractable for quantum computers. In a yes-instance execution of Protocol \ref{protocol:main-protocol-informal}, and for any classical probabilistic (resp. quantum) polynomial-time verifier interacting with the honest prover, there exists a classical probabilistic polynomial-time (resp. quantum polynomial-time) simulator such that the simulator's output is classical (resp. quantum) computationally indistinguishable from that of the verifier.
\end{lemma}

The reason we delay the verifier's reveal of $r_v$ (rather than completing the coin-flipping in one step, as is done in the protocol in \cite{qma}) is fairly easily explained. In our classical-verifier protocol, the prover cannot physically send the quantum state $E_K(\rho)$ to its verifier before the random string $r$ is decided, as the prover of the protocol in \cite{qma} does. If we allow our prover to know $r$ at the time when it performs measurements on the witness $\rho$, it will trivially be able to cheat.

The trapdoor check, meanwhile, is an addition which we make because we wish to construct a \emph{classical} simulator for our protocol when we prove that it is zero-knowledge. Since our verifier is classical, we need to achieve a classical simulation of the protocol in order to prove that its execution (in yes-instances) does not impart to the verifier any knowledge it could not have generated itself. During the measurement protocol, however, the prover is required to perform quantum actions which no classical polynomial-time algorithm could simulate unless it had access to the verifier's function trapdoors. Naturally, we cannot ask the verifier to reveal its trapdoors before the measurement protocol takes place. As such, we ask the verifier to reveal them immediately afterwards instead, and show in Section \ref{section:zk} that this (combined with the encryption-scheme properties of the prover's encoding $E_K$) allows us to construct a classical simulator for Protocol \ref{protocol:main-protocol-informal} in yes-instances.

The organisation of the paper is as follows.
\begin{enumerate}[itemsep=1pt, topsep=0pt]
\item Section 2 (`Ingredients') outlines the other protocols which we use as building blocks.
\item Section 3 (`The protocol') introduces our argument system for QMA.
\item Section 4 (`Completeness of protocol') gives a completeness lemma for the argument system introduced in section 3.
\item Section 5 (`Soundness of protocol') proves that the argument system introduced in section 3 is sound against quantum polynomial-time provers.
\item Section 6 (`Zero-knowledge property of protocol') proves that the argument system is zero-knowledge (that yes-instance executions can be simulated classically).
\end{enumerate}

\begin{remark}
\label{rk:2pc}
As Broadbent et al.\ note in~\cite[Section 1.3]{qma}, argument systems can often be made zero-knowledge by employing techniques from secure two-party computation (2PC). The essential idea of such an approach, applied to our particular problem, is as follows: the prover and the verifier would jointly simulate the classical verifier of the \cite{measurement} measurement protocol using a (classical) secure two-party computation protocol, and zero-knowledge would follow naturally from simulation security. (This technique is similar in spirit to those which are used in \cite{ip-in-czk} to show that any classical-verifier interactive \emph{proof} system can be made zero-knowledge.) We think that the 2PC approach applied to our problem would have many advantages, including that it is more generally applicable than our approach; however, we also believe that our approach is a more direct and transparent solution to the particular problem at hand, and that it provides an early example of how two important results might be fruitfully combined. As such, we expect that our approach may more easily lead to extensions and improvements.
\end{remark}

\paragraph{Related work.} Subsequent to the completion of this work, there have been several papers which explore other extensions and applications of the argument system from~\cite{measurement}, and also papers which propose  zero-knowledge protocols (with different properties from ours) for QMA. Many of these works focus on decreasing the amount of interaction required to implement a proof or argument system for QMA. Although none of these works directly builds on or supersedes ours, we review them briefly for the reader's convenience. In the category of extensions on the work of \cite{measurement}, we mention~\cite{alagic2019non}, which proposes a \emph{non-interactive} zero-knowledge variant of the Mahadev protocol and proves its security in the quantum random oracle model. (Of course, our protocol is interactive, and our analysis holds in the standard model.) In the category of `short' proof and argument systems for QMA, we mention three independent works. In~\cite{bitansky2019post}, the authors present a constant-round computationally zero-knowledge argument system for QMA. In~\cite{broadbent2019zero} and~\cite{coladangelo2019non} the authors present non-interactive zero-knowledge proof and argument systems, respectively, for QMA, with different types of setup phases. The main difference between all three of these new protocols and our protocol is that the three protocols mentioned all involve the exchange of quantum messages (although, in~\cite{coladangelo2019non}, only the setup phase requires quantum communication).

\emph{Acknowledgments.} We thank Zvika Brakerski, Andru Gheorghiu, and Zhengfeng Ji for useful discussions. We thank an anonymous referee for suggesting the approach based on secure 2PC sketched in Remark~\ref{rk:2pc}. 
Thomas Vidick is supported by NSF CAREER Grant CCF-1553477, AFOSR YIP award number FA9550-16-1-0495, MURI Grant FA9550-18-1-0161, a CIFAR Azrieli Global Scholar award, and the IQIM, an NSF Physics Frontiers Center (NSF Grant PHY-1125565) with support of the Gordon and Betty Moore Foundation (GBMF-12500028). Tina Zhang acknowledges support from the Richard G. Brewer Prize and Caltech's Ph11 program.

\section{Ingredients}
\label{section:ingredients}
The protocol we present in section \ref{section:protocol} combines techniques which were introduced in prior works for the design of protocols to solve related problems. In this section, we outline these protocols in order to introduce notation and groundwork which will prove useful in the remainder of the paper. We also provide formal definitions of QMA and of zero-knowledge.

\subsection{Definitions}

\begin{definition}[QMA]

The following definition is taken from \cite{qma}.

  A promise problem $A = (A_{yes},A_{no})$ is contained in the complexity
  class $\mathrm{QMA}_{\alpha,\beta}$ if there exists a polynomial-time generated
  collection
  \begin{equation}
    \bigl\{V_x\,:\,x\in A_{yes} \cup A_{no}\bigr\}
  \end{equation}
  of quantum circuits and a polynomially bounded function $p$ possessing the
  following properties:
  \begin{enumerate}
  \item[1.]
    For every string $x\in A_{yes}\cup A_{no}$, one has that $V_x$ is a
    measurement circuit taking $p(\abs{x})$ input qubits and outputting
    a single bit.
  \item[2.]
    \emph{Completeness}.
    For all $x\in A_{yes}$, there exists a $p(\abs{x})$-qubit state $\rho$
    such that $\Pr(V_x(\rho) = 1) \geq \alpha$.
  \item[3.]
    \emph{Soundness}.
    For all $x\in A_{no}$, and every $p(\abs{x})$-qubit state
    $\rho$, it holds that $\Pr(V_x(\rho) = 1) \leq \beta$.
  \end{enumerate}
\end{definition}

In this definition, $\alpha,\beta\in[0,1]$ may be constant values or functions
of the length of the input string~$x$.
When they are omitted, it is to be assumed that they are $\alpha = 2/3$ and
$\beta = 1/3$.
Known error reduction methods~\cite{kitaev2002classical,marriott2005quantum} imply that a wide
range of selections of $\alpha$ and $\beta$ give rise to the same complexity
class.
In particular, $\mathrm{QMA}$ coincides with $\mathrm{QMA}_{\alpha,\beta}$ for
$\alpha = 1 - 2^{-q(\abs{x})}$ and $\beta = 2^{-q(\abs{x})}$, for any
polynomially bounded function $q$.

\begin{definition}[Zero-knowledge]
  Let $(P, V)$ be an interactive proof system (with a classical verifier $V$) for a promise problem $A = (A_{yes}, A_{no})$. Assume that (possibly among other arguments) $P$ and $V$ both take a problem instance $x \in \{0,1\}^*$ as input. $(P,V)$ is
  \emph{computational zero-knowledge} if, for every probabilistic polynomial-time (PPT) $V^*$, there exists a polynomial-time generated
 simulator $S$ such that, when $x \in A_{yes}$, the distribution of $V^*$'s final output after its interaction with the honest prover $P$ is \emph{computationally indistinguishable} from $S$'s output distribution. More precisely, let $\lambda$ be a security parameter, let $n$ be the length of $x$ in bits, and and let $\{D_{n,\lambda}\}_{n, \lambda}$ and $\{S_{n,\lambda}\}_{n, \lambda}$ be the two distribution ensembles representing, respectively, the verifier $V^*$'s output distribution after an interaction with the honest prover $P$ on input $x$, and the simulator's output distribution on input $x$. If $(P,V)$ is computationally zero-knowledge, we require that, for all PPT algorithms $A$, the following holds:
 
 \begin{equation*}
 \Big| \underset{y \leftarrow D_{n,\lambda}}{\mathrm{Pr}}[A(y) = 1] - \underset{y \leftarrow S_{n,\lambda}}{\mathrm{Pr}}[A(y) = 1] \Big| = \mu(n) \nu(\lambda),
 \end{equation*}
 
 where $\mu(\cdot)$ and $\nu(\cdot)$ are negligible functions.

  \label{def_qczk}
\end{definition}

\subsection{Single-qubit-verifier proof system for QMA (\cite{xz})}

Morimae and Fitzsimons (\cite{xz}) present a proof system for languages (or promise problems) in QMA whose verifier is classical except for a single-qubit quantum register, and which is sound against arbitrary quantum provers. The proof system relies on the QMA-completeness of the 2-local XZ Hamiltonian problem, which is defined as follows.

\begin{definition}[2-local XZ Hamiltonian (promise) problem]
\label{def:2-local-xz-problem}~\\
\emph{Input.} An input to the problem consists of a tuple $x = (H, a, b)$, where
\begin{enumerate}[itemsep=1pt, topsep=0pt]
\item $H = \sum_{s=1}^S d_s H_s$ is a Hamiltonian acting on $n$ qubits, each term $H_s$ of which 
\begin{enumerate}[itemsep=1pt, topsep=0pt]
\item has a weight $d_s$ which is a polynomially bounded rational number,
\item satisfies $0 \leq H_s \leq I$,
\item acts as the identity on all but a maximum of two qubits,
\item acts as the tensor product of Pauli observables in $\{\sigma_X, \sigma_Z\}$ on the qubits on which it acts nontrivially.
\end{enumerate}
\item $a$ and $b$ are two real numbers such that
\begin{enumerate}[itemsep=1pt, topsep=0pt]
\item $a < b$, and
\item $b - a = \Omega(\frac{1}{\textsf{poly}(|x|)})$. 
\end{enumerate}
\end{enumerate}

\emph{Yes:} There exists an $n$-qubit state $\sigma$ such that $\big< \sigma, H \big> \leq a.$\footnote{The angle brackets $\big< \cdot , \cdot \big>$ denote an inner product between two operators which is defined as follows: $\big< A, B \big> = \mathrm{Tr}(A^* B)$ for any $A, B \in \mathrm{L}(\mathcal{X}, \mathcal{Y})$, where the latter denotes the space of linear maps from a Hilbert space $\mathcal{X}$ to a Hilbert space $\mathcal{Y}$.}

\emph{No:} For every $n$-qubit state $\sigma$, it holds that $\big< \sigma, H \big> \geq b.$
\end{definition}

\begin{remark}
Given a Hamiltonian $H$, we call any state $\sigma^*$ which causes $\big< \sigma^*, H \big>$ to take its minimum possible value a \emph{ground state} of $H$, and we refer to the value $\big< \sigma^*, H \big>$ as the \emph{ground energy} of $H$. \\
\end{remark}

The following theorem is proven by Biamonte and Love in \cite[Theorem 2]{xz-qma-complete}.

\begin{theorem}\label{thm:xz-complete}
The 2-local XZ Hamiltonian problem is complete for QMA.
\end{theorem}

We now describe an amplified version of the protocol presented in \cite{xz}, and give a statement about its completeness and soundness which we will use. (See \cite{xz} for a more detailed presentation of the unamplified version of this protocol.)

\begin{protocol}[Amplified variant of the single-qubit-verifier proof system for QMA from~\cite{xz}]
\label{protocol:mf16}
\quad\\

\emph{Notation.}
Let $L = (L_{yes}, L_{no})$ be any promise problem in QMA; let $x\in\{0,1\}^*$ be an input; and let $(H,a,b)$ be the instance of the 2-local XZ Hamiltonian problem to which $x$ reduces.
\begin{enumerate}[itemsep=1pt, topsep=0pt]
\item If $x \in L_{yes}$, the ground energy of $H$ is at most $a$.
\item if $x \in L_{no}$, the ground energy of $H$ is at least $b$.
\item $b - a \geq \frac{1}{\textsf{poly}(|x|)}$.
\end{enumerate}
Let $H = \sum_{s=1}^S d_s H_s$, as in Definition \ref{def:2-local-xz-problem}. Define
\begin{gather*}
\pi_s = \frac{|d_s|}{\sum_s |d_s|}\;.
\end{gather*}
\emph{Parties.}
The proof system involves 
\begin{enumerate}[itemsep=1pt, topsep=0pt]
\item A \emph{verifier}, who implements a classical probabilistic polynomial-time procedure with access to a one-qubit quantum register; and
\item A \emph{prover}, who is potentially unbounded, but whose honest behaviour in yes-instances can be implemented in quantum polynomial time. 
\end{enumerate}
The verifier and the prover communicate quantumly.

\emph{Inputs.}
\begin{enumerate}[itemsep=1pt, topsep=0pt]
\item Input to the verifier: the Hamiltonian $H$ and the numbers $a$ and $b$.
\item Input to the prover: the Hamiltonian $H$, the numbers $a$ and $b$, and the quantum state $\rho = \sigma^{\otimes m}$, where $\sigma$ is a ground state of the Hamiltonian $H$.
\end{enumerate}

\emph{Protocol.}
\begin{enumerate}[itemsep=1pt, topsep=0pt]
\item The verifier selects uniformly random coins $r=(r_1,\ldots,r_m)$. 
\item 
For each $j \in \{1, \ldots, m\}$, the verifier uses $r_j$ to select a random $s_j\in\{1,\ldots,S\}$ according to the distribution $D$ specified as follows:
\begin{gather*}
D(s) = \pi_s, \quad \text{for } s \in \{1, \ldots, S\}\;.
\end{gather*}
\item The prover sends a state $\rho$ to the verifier one qubit at a time. (The honest prover sends the state $\sigma^{\otimes m}$ that consists of $m$ copies of the ground state of $H$.)
\item The verifier measures $H_{s_j}$ for $j = 1, \ldots, m$, taking advantage of the fact that—if the prover is honest—it is given $m$ copies of $\sigma$. (`Measuring $H_{s_j}$', in this case, entails performing at most two single-qubit measurements, in either the standard or the Hadamard basis, on qubits in $\rho$, and then computing the product of the two measurement outcomes.)
\item The verifier initialises a variable \textsc{Count} to 0. For each $j\in\{1,\ldots,m\}$, if the $j$th product that it obtained in the previous step was equal to $-\mathrm{sign}(d_j)$, the verifier adds one to \textsc{Count}.
\item If $\frac{\textsc{Count}}{m}$ is closer to $\frac{1}{2} - \frac{a}{\sum_s 2 |d_s|}$ than to $\frac{1}{2} - \frac{b}{\sum_s 2 |d_s|}$, the verifier accepts. Otherwise, it rejects. 
\end{enumerate}
\end{protocol}

\begin{claim}
\label{claim:proof-amplification-works}
Given an instance $x=(H,a,b)$ of the $2$-local XZ Hamiltonian problem, there is a polynomial $P$ (depending only on $a$ and $b$) such that, for any $m = \Omega(P(|x|))$, the following holds.
In a yes-instance, the procedure of Protocol \ref{protocol:mf16} accepts the state $\rho = \sigma^{\otimes m}$ with probability exponentially close (in $|x|$) to 1. In a no-instance, the probability that it accepts any state is exponentially small in $|x|$. 
\end{claim}

\begin{proof}
Consider the probability (over the choice of $r_j$ and the randomness arising from measurement) that the $j$th measurement from step 4 of Protocol \ref{protocol:mf16}, conditioned on previous measurement outcomes,  yields $-\mathrm{sign}(d_j)$. Denote this probability by $q_j$.

As shown in~\cite[Section IV]{proof-of-gap}, it is not hard to verify that 
\begin{enumerate}[itemsep=1pt, topsep=0pt]
\item when $x \in L$, if the prover sends the honest witness $\sigma^{\otimes m}$, then $q_j \geq \frac{1}{2} - \frac{a}{\sum_s 2|d_s|}$, and
\item when $x \notin L$, for any witness that the prover sends, $q_j \leq \frac{1}{2} - \frac{b}{\sum_s 2|d_s|}$.
\end{enumerate}
The difference between the two cases is inverse polynomial in the size of the input to the 2-local XZ Hamiltonian problem. It is straightforward to show that, for an appropriate choice of $m$, this inverse polynomial gap can be amplified to an exponential one: see Appendix \ref{appendix:chernoff-bound}.
\end{proof}

\begin{remark}
\label{remark:rho_r}
It will be useful later to establish at this point that, if the string $r$ from step 1 of Protocol \ref{protocol:mf16} is fixed, it is simple to construct a state $\rho_r$ which will pass the challenge determined by $r$ with probability 1. One possible procedure is as follows.
\begin{enumerate}[itemsep=1pt, topsep=0pt]
\item For each $j \in 1, \ldots, m$:

Suppose that $H_{s_j} = d_j P_1 P_2$, and that $P_1,P_2\in\{\sigma_X,\sigma_Z\}$ act on qubits $\ell_1$ and $\ell_2$, respectively.
\begin{enumerate}[itemsep=1pt, topsep=0pt]
\item If $-\mathrm{sign}(d_j) = 1$, initialise the $((j-1)n + \ell_1)$th qubit to the +1 eigenstate of $P_1$, and likewise, initialise the $((j-1)n + \ell_2)$th qubit to the +1 eigenstate of $P_2$.
\item If $-\mathrm{sign}(d_j) = -1$, initialise the $((j-1)n + \ell_1)$th qubit to the +1 eigenstate of $P_1$, and initialise the $((j-1)n + \ell_2)$th qubit to the $-1$ eigenstate of $P_2$.
\end{enumerate}
\item Initialise all remaining qubits to $\ket{0}$.
\end{enumerate}
It is clear that the $\rho_r$ produced by this procedure is a tensor product of $\ket{0}$, $\ket{1}$, $\ket{+}$ and $\ket{-}$ qubits.
\end{remark}

\subsection{Measurement protocol (\cite{measurement})}
\label{section:mah18}
In \cite{measurement}, Mahadev presents a measurement protocol between a quantum prover and a classical verifier which, intuitively, allows the verifier to obtain trustworthy standard and Hadamard basis measurements of the prover's quantum state from purely classical interactions with it. The soundness of the measurement protocol relies upon the security properties of functions that \cite{measurement} terms \emph{noisy trapdoor claw-free functions} and \emph{trapdoor injective functions}, of which Mahadev provides explicit constructions presuming upon the hardness of LWE. (A high-level summary of these constructions can be found in Appendix \ref{appendix:trapdoor-check}.) Here, we summarise the steps of the protocol, and state the soundness property that it has which we will use.

\begin{protocol}[Classical-verifier, quantum-prover measurement protocol from \cite{measurement}]
\label{protocol:mah18}~

\emph{Parties.}
The proof system involves
\begin{enumerate}[itemsep=1pt, topsep=0pt]
\item A \emph{verifier}, which implements a classical probabilistic polynomial-time procedure; and
\item A \emph{prover}, which implements a quantum polynomial-time procedure. 
\end{enumerate}
The verifier and the prover communicate classically.

\emph{Inputs.}
\begin{enumerate}[itemsep=1pt, topsep=0pt]
\item Input to the prover: an $n$-qubit quantum state $\rho$, whose qubits the verifier will attempt to derive honest measurements of in the standard and Hadamard bases.
\item Input to the verifier:
\begin{enumerate}[itemsep=1pt, topsep=0pt]
\item A string $h \in \{0,1\}^n$, which represents the bases (standard or Hadamard) in which it will endeavour to measure the qubits of $\rho$. $h_i = 0$ signifies that the verifier will attempt to obtain measurement outcomes of the $i$th qubit of $\rho$ in the standard basis, and $h_i = 1$ means that the verifier will attempt to obtain measurement outcomes of the $i$th qubit of $\rho$ in the Hadamard basis.
\item An \emph{extended trapdoor claw-free function family} (ETCFF family), as defined in Section 4 of \cite{measurement}. The description of an ETCFF family specifies a large number of algorithms, and we do not attempt to enumerate them. Instead, we proceed to describe the verifier's prescribed actions at a level of detail which we believe to be sufficient for our purposes, and refer the reader to \cite{measurement} for a finer exposition.
\end{enumerate}
\end{enumerate}

\emph{Protocol.} 
\begin{enumerate}[itemsep=1pt, topsep=0pt]
\item For each $i \in \{1, \ldots, n\}$ (see `Inputs' above for the definition of $n$), the verifier generates an ETCFF function key $\kappa_i$ using algorithms provided by the ETCFF family, along with a trapdoor $\tau_{\kappa_i}$ for each function, and sends all of the keys $\kappa$ to the prover. It keeps the trapdoors $\tau$ to itself. If $h_i = 0$, the $i$th key $\kappa_i$ is a key for an \emph{injective} function $g$, and if $h_i = 1$, it is a key for a \emph{two-to-one} function $f$ known as a `noisy trapdoor claw-free function'. Intuitively, the $g$ functions are one-to-one trapdoor one-way functions, and the $f$ functions are two-to-one trapdoor collision-resistant hash functions. The keys for 
 $f$ functions and those for $g$ functions are computationally indistinguishable. (For convenience, we will from now on refer to the function specified by $\kappa_i$ either as $f_{\kappa_i}$ or as $g_{\kappa_i}$. Alternatively, we may refer to it as $\eta_{\kappa_i}$ if we do not wish to designate its type.\footnote{The letter $\eta$ has been chosen because it bears some resemblance to the Latin letter $h$.}) A brief outline of how these properties are achieved using LWE is given in Appendix~\ref{appendix:trapdoor-check}.
 
We make two remarks about the functions $\eta_{\kappa_i}$ which will become relevant later. 
\begin{enumerate}[itemsep=1pt, topsep=0pt]
\item The functions $\eta_{\kappa_i}$ always have domains of the form $\{0,1\} \times \mathcal{X}$, where $\mathcal{X} \subseteq \{0,1\}^w$ for some length parameter $w$.
\item The outputs of both the $f$ and the $g$ functions should be thought of not as strings but as \emph{probability distributions}. The trapdoor $\tau_{\kappa_i}$ inverts the function specified by $\kappa_i$ in the sense that, given a sample $y$ from the distribution $Y = \eta_{\kappa_i}(b \| x)$, along with the trapdoor $\tau_{\kappa_i}$, it is possible to recover $b \| x$, as well as any other $b' \| x'$ which also maps to $Y$ under $\eta_{\kappa_i}$ (should it exist).
\end{enumerate}

\begin{definition}
\label{def:deterministic-etas}
Suppose that $\eta_{\kappa_i}$ is the function specified by $\kappa_i$, whose output on each input $b \| x$ in its domain $\{0,1\} \times \mathcal{X}$ is a probability distribution $Y$. Define a (deterministic) function $\eta^*_{\kappa_i}(b\|x, e)$ which takes as input an $b\|x \in \{0,1\} \times \mathcal{X}$ and a randomness $e \in \mathcal{E}$, for some well-defined finite set $\mathcal{E}$, and returns a sample $y_e$ from the distribution $Y = \eta_{\kappa_i}(b \| x)$.
\end{definition}

\begin{definition}
\label{def:preimage-of-y}
Let $\eta_{\kappa_i}$ be the function specified by $\kappa_i$, with domain $\{0,1\} \times \mathcal{X}$. Let $y$ be a sample from one of the distributions $Y \in \mathcal{Y}$, where $\mathcal{Y}$ is the range of $\eta_{\kappa_i}$. It is guaranteed that the distributions in the range of $\eta_{\kappa_i}$ have compact support, and we call $b\|x \in \{0,1\} \times \mathcal{X}$ a \emph{preimage of $y$} if $y$ is in the support of the probability distribution $\eta_{\kappa_i}(b\|x)$.
\end{definition}

\item The prover uses the function keys $\kappa_1, \ldots, \kappa_n$ to `commit' to the quantum state of each of the $n$ qubits of $\rho$, and sends $n$ commitment strings $y_1, \ldots, y_n$ to the verifier. We direct the reader to Sections 2.2 and 5.1 of \cite{measurement} for a high-level, and then a more detailed, description of the commitment procedure, and explanations of how such a procedure will allow the verifier to extract trustworthy standard and Hadamard basis measurements of the qubits of $\rho$. For convenience, we summarise the procedure without justification here.

For each $i$, the prover concatenates to the qubit indexed by $i$ (which we call, following \cite{measurement}, the $i$th `committed qubit', and whose state we denote by\footnote{Strictly speaking, the state of the $i$th commited qubit may not be a pure state, but we ignore this fact for simplicity of presentation.} $\ket{\psi_i} = \gamma_i \ket{0} + \delta_i \ket{1}$) a register (the `preimage register') containing a uniform superposition over all $x \in \mathcal{X}$ (where $\{0,1\} \times \mathcal{X}$ is the domain of the function $\eta_{\kappa_i}$). It also concatenates to it a register containing a superposition over all $e \in \mathcal{E}$, with $\mathcal{E}$ defined as it is in Definition \ref{def:deterministic-etas}. It then applies the function $\eta^*_{\kappa_i}$ (see Definition \ref{def:deterministic-etas}) in superposition to $\sum_x \ket{\psi_i} \ket{x} \ket{e}$, and collects the outputs in a new register (the `output register'), obtaining the quantum state
\begin{gather*}
\sum_{x \in \{0,1\}^w} \gamma_i \ket{0} \ket{x} \ket{e}\ket{\eta^*_{\kappa_i}(0 \| x, e)} + \delta_i \ket{1} \ket{x} \ket{e} \ket{\eta^*_{\kappa_i}(1 \| x, e)}\;.
\end{gather*}
Finally, the prover measures the output register, obtains a sample $y_i$, and sends $y_i$ to the verifier as its $i$th commitment string.
\item The verifier then chooses at random to run either a \emph{test round} or a \emph{Hadamard round}. Each option is chosen with probability $\frac{1}{2}$.
\begin{enumerate}[itemsep=1pt, topsep=0pt]
	\item If a test round is chosen, the verifier requests standard basis measurements of each committed qubit $\ket{\psi_i}$ and its associated preimage register (recall that, in the previous step, only the output registers were measured), and receives measurement results $\beta_i, x_i$ for $i \in \{1, \ldots, n\}$ from the prover. It then checks, for each $i \in \{1, \ldots, n\}$, whether the $i$th measurement result is a preimage of $y_i$. (See Definition \ref{def:preimage-of-y} for a definition of the phrase `preimage of $y$'.) More precisely: if $h_i = 0$, the verifier accepts iff $y_i$ is in the support of $g_{\kappa_i}(\beta_i \| x_i)$ (recalling that $g_{\kappa_i}(\beta_i \| x_i)$ is a probability distribution); and if $h_i = 1$, the verifier accepts iff $y_i$ is in the support of $f_{\kappa_i}(\beta_i \| x_i)$.
	\item If a Hadamard round is chosen, the verifier requests Hadamard basis measurements of each committed qubit and its associated preimage register, and receives $\beta_i, d_i$ for $i \in \{1, \ldots, n\}$ from the prover. (We note, for future reference, that the distribution of the strings $d_i$ sent by the honest prover is the uniform distribution over $\{0,1\}^w$.) It then obtains the measurement results it desires using the procedure outlined in Section 2.2 of \cite{measurement}:
	\begin{enumerate}[itemsep=1pt, topsep=0pt]
	\item If $h_i = 0$, the verifier ignores the measurement results $\beta_i, d_i$, and uses the trapdoor $\tau_{\kappa_i}$ to invert $y_i$. If the inverse of $y_i$ exists, the inversion procedure outputs a preimage $(m_i, x_{m_i, y_i})$, and the verifier stores $m_i$ as the standard-basis measurement outcome. If the inverse does not exist, the verifier stores a random bit as $m_i$ and rejects. The specification that the verifier stores a random bit is made so that the $n$-bit string of measurement results, $m$, is well-defined even when the verifier rejects.
	\item If $h_i = 1$, the verifier computes the two preimages $x_{0, i}$ and $x_{1, i}$ to $y_i$ under $f_{\kappa_i}$ using the trapdoor $\tau_{\kappa_i}$. It then stores $\beta_i \oplus d_i \cdot (x_{0,i} \oplus x_{1,i})$ as the Hadamard basis measurement outcome. Once again, if either inverse does not exist, the verifier stores a random bit as $m_i$ and rejects.
	\end{enumerate}
\end{enumerate}
\end{enumerate}
\end{protocol}

The soundness property of this measurement protocol (Claim 7.1 in \cite{measurement}) is stated below. Before we state it, we briefly introduce the notation used in \cite{measurement}. (For a more detailed exposition, we refer the reader to Section 5.3 of \cite{measurement}.) $D_{\mathbb{P},h}$ refers to the distribution over measurement results $m \in \{0,1\}^n$ that the verifier obtains when it executes a Hadamard round with the prover labelled $\mathbb{P}$ on the basis choice $h$. $D^C_{\mathbb{P},h}$ is the same distribution, but conditioned on the verifier accepting (in a Hadamard round). $D_{\xi, h}$ is the distribution over measurement outcomes in $\{0,1\}^n$ that would result from directly measuring the quantum state $\xi$ in the bases determined by $h$. $p_{h, T}$ and $p_{h, H}$ are defined so that the verifier's probability of accepting (on basis choice $h$) in a test and a Hadamard round, respectively, are $1 - p_{h,T}$ and $1 - p_{h,H}$. $\| \cdot \|_{TV}$ denotes the total variation norm, and $A \approx_c B$ indicates that two distributions $A$ and $B$ are (quantum) computationally indistinguishable.

\begin{claim}\label{claim:measurement-protocol}
Assume that the Learning With Errors problem (with the same choices of parameters as those made in \cite[Section 9]{measurement}) is quantum computationally intractable. Then, for any arbitrary quantum polynomial-time prover $\mathbb{P}$ who executes the measurement protocol (Protocol \ref{protocol:mah18}) with the honest verifier $V$, there exists a quantum state $\xi$, a prover $\mathbb{P}'$ and a negligible function $\mu$ such that
\begin{gather*}
\|D^C_{\mathbb{P}, h} - D_{\mathbb{P'},h}\|_{TV} \leq \sqrt{p_{h,T}} + p_{h,H} + \mu \quad \text{and} \\
D_{\mathbb{P}',h} \approx_c D_{\xi,h}\;. \label{eq:comp-indist}
\end{gather*}	
\end{claim}

\subsection{Zero-knowledge proof system for QMA (\cite{qma})}
\label{section:bjsw16}

In \cite{qma}, Broadbent, Ji, Song and Watrous describe a protocol involving a quantum polynomial-time verifier and an unbounded prover, interacting quantumly, which constitutes a zero-knowledge proof system for promise problems in QMA. (Although it is sound against arbitrary provers, the system in fact only requires an honest prover to perform quantum polynomial-time computations.) We summarise the steps of their protocol below. For details and fuller explanations, we refer the reader to~\cite[Section 3]{qma}.

\begin{protocol}[Zero-knowledge proof system for QMA from \cite{qma}]
\label{protocol:bjsw16}~

\emph{Notation.}
Let $L$ be any promise problem in QMA. For a definition of the \emph{$k$-local Clifford Hamiltonian problem}, see~\cite[Section 2]{qma}. The $k$-local Clifford Hamiltonian problem is QMA-complete for $k=5$; therefore, for all possible inputs $x$, there exists a 5-local Clifford Hamiltonian $H$ (which can be computed efficiently from $x$) whose terms are all operators of the form $C^* \ket{0^k}\bra{0^k} C$ for some Clifford operator $C$, and such that
\begin{enumerate}[itemsep=1pt, topsep=0pt]
\item if $x \in L_{yes}$, the ground energy of $H$ is $\leq 2^{-p}$,
\item if $x \in L_{no}$, the ground energy of $H$ is $\geq \frac{1}{q}$, 
\end{enumerate}
for some positive integers $p$ and $q$.

\emph{Parties.}
The proof system involves 
\begin{enumerate}[itemsep=1pt, topsep=0pt]
\item A \emph{verifier}, who implements a quantum polynomial-time procedure;
\item A \emph{prover}, who is unbounded, but who is only required by the protocol to implement a quantum polynomial-time procedure.
\end{enumerate}
The verifier and the prover communicate quantumly.

\emph{Inputs.} 
\begin{enumerate}[itemsep=1pt, topsep=0pt]
\item Input to the verifier:
\begin{enumerate}[itemsep=1pt, topsep=0pt]
\item The Hamiltonian $H$.
\item A quantum computationally concealing, perfectly binding (classical) commitment protocol.
\item A proof system for NP sound against arbitrary quantum provers.
\end{enumerate}
\item Input to the prover:
\begin{enumerate}[itemsep=1pt, topsep=0pt]
\item The Hamiltonian $H$.
\item The $n$-qubit quantum state $\rho$, where $\rho$ is a ground state of the Hamiltonian $H$.
\item A quantum computationally concealing, perfectly binding (classical) commitment protocol.
\item A zero-knowledge proof system for NP sound against arbitrary quantum provers.
\end{enumerate}
\end{enumerate}

\emph{Protocol.} 
\begin{enumerate}[itemsep=1pt, topsep=0pt]
\item \emph{The prover's encoding step}. The prover firstly applies a concatenated Steane code (which maps every one qubit to $N$ qubits) to each qubit in $\rho$. (For details on the concatenated Steane code, see~\cite[Appendix A.6]{qma}. It will be important to Broadbent et al.'s purposes—and ours, in Section \ref{section:protocol}—that this code admits transversal applications of Clifford operations.) It then executes the following steps:
\begin{enumerate}[itemsep=1pt, topsep=0pt]
\item Concatenate $N$ trap qubits to the end of each logical qubit (alternatively, to the end of each $N$-qubit block) in the result of applying the concatenated Steane code to $\rho$. Each trap qubit is initialised uniformly at random to one of $\ket{0}, \ket{+}, \ket{+_y}$.\footnote{$\ket{+_y}$ here refers to the state $\frac{1}{\sqrt{2}}(\ket{0} + i \ket{1})$.} Denote the string that records the choices of trap qubits for all $n$ logical qubits by $t = t_1, \ldots, t_n$.
\item Permute each $2N$-tuple of qubits in the result of (a) according to a uniformly random permutation $\pi \in S_{2N}$. (Note that the same permutation $\pi$ is applied to every $2N$-tuple.)
\item Apply a Pauli one-time pad $X^a Z^b$, for uniformly random $a, b \in \{0,1\}^{2nN}$, to the entire $2nN$-qubit state.
\end{enumerate}
\begin{remark}
The number $N$ functions as a security parameter, both for this protocol and for the protocol that we construct in section \ref{section:protocol}. Accordingly, we assume that $N$ is polynomial in the size of the verifier's input.
\end{remark}

The prover's encoding applied to $\rho$ is denoted by $E(\rho)$, and the procedure $E$ is fully determined by the encoding key $(t, \pi, a, b)$ which the prover chose to use. At this point, the prover sends the state $E(\rho)$ to the verifier, along with a commitment (using some perfectly binding, computationally concealing classical commitment protocol) to the tuple $(\pi, a, b)$. (A commitment to the sequence of trap qubits $t$ is unnecessary because, in a sense, the trap qubits exist only to check the verifier.) Let the prover's commitment string be denoted $z$.
\item \emph{Coin-flipping protocol.} The prover and the verifier execute a coin-flipping protocol, choosing a string $r$ of fixed length uniformly at random. This random string $r$ determines a local Hamiltonian term $H_r = C^*_r \ket{0^k} \bra{0^k} C_r$ that is to be tested. (This step can be implemented, of course, using the same classical commitment protocol that the prover employed in the previous step.)
\item \emph{Verifier's challenge.} The verifier applies the Clifford $C_r$ transversally to the qubits on which the $k$-local Hamiltonian term $H_r$ acts nontrivially, and measures them in the standard basis. It then sends the measurement results $u_{i_1}, \ldots, u_{i_k}$ which it obtained to the prover. (Each $u_i$ is a $2N$-bit string, and $i_1, \ldots, i_k$ are the indices of the logical qubits on which the term $H_r$ acts nontrivially.)
\item \emph{Prover's response}. The prover receives the verifier's measurement results $u$, and firstly checks whether they cause a predicate $Q(t, \pi, a, b, r, u)$ to be satisfied. (We will explain the predicate $Q$ in more detail shortly. Intuitively, $Q$ is satisfied if and only if both verifier \emph{and} prover behaved honestly.) If $Q$ is not satisfied, the prover aborts, causing the verifier to reject. If $Q$ is satisfied, then the prover proves to the verifier, using an NP zero-knowledge protocol, that there exists randomness $s_p$ and an encoding key $(t, \pi, a, b)$ such that $z = \textsf{commit}((\pi, a, b), s_p)$ and $Q(t, \pi, a, b, r, u) = 1$.

We now describe the predicate $Q$ in precise terms. For convenience, Broadbent et al.\ define a predicate $R_r$, which represents the prover's check after it reverses the effects on $u$ of the one-time pad $X^a Z^b$, and then proceed to define $Q$ in terms of $R_r$. Since we will shortly have cause to alter the definition of $R_r$, we quote the definition of $R_r$ used in \cite{qma} for comparison. (Note that we have altered their notation slightly: the strings that they call $y_i$ we have called $q_i$, to avoid a conflict.)

\begin{definition}[Definition of $R_r$]
\label{def:original-R_r}

Let $u_{i_1}, \ldots, u_{i_k}$ be the measurement results that the verifier sent to the prover in step 3. For each $i \in \{i_1, \ldots, i_k\}$, define strings $q_i, z_i$ in $\{0,1\}^N$ such that $\pi(q_i \| z_i) = u_i$ (alternatively: $\pi^{-1}(u_i) = q_i \| z_i)$. The predicate $R_r(t, \pi, u)$ takes the value 1 if and only if the following two conditions are met:
\begin{enumerate}[itemsep=1pt, topsep=0pt]
\item $q_i \in \mathcal{D}_N$ for every $i \in \{i_1, \ldots, i_k\}$, and $q_i \in \mathcal{D}^1_N$ for at least one index $i \in \{i_1, \ldots, i_k\}$. ($\mathcal{D}_N = \mathcal{D}^0_N \cup \mathcal{D}^1_N$ is the set of all valid classical $N$-bit codewords of the concatenated Steane code. We note, as an aside, that $\mathcal{D}^0_N$ and $\mathcal{D}^1_N$ are both sets of size $8^t$, where $t$ is an even positive integer such that $7^t = N$; as such, $\mathcal{D}_N$ is polynomially sized.)
\item $\bra{z_{i_1} \cdots z_{i_k}} C^{\otimes n}_r \ket{t_{i_1} \cdots t_{i_k}} \neq 0.$
\end{enumerate}
\end{definition}

Now we define the predicate $Q(t, \pi, a, b, r, u)$ in terms of $R_r$:

\begin{definition}[Definition of $Q$]
\label{def:original-Q}

Let $c_1, \ldots, c_n, \: d_1, \ldots, d_n \in \{0,1\}^{2N}$ be the unique strings such that
\begin{gather*}
C_r^{\otimes 2N}(X^{a_1} Z^{b_1} \otimes \cdots \otimes X^{a_n} Z^{b_n}) = \alpha(X^{c_1} Z^{d_1} \otimes \cdots \otimes X^{c_n} Z^{d_n}) C_r^{\otimes 2N}
\end{gather*}
for some $\alpha \in \{1, i, -1, -i\}$. (It is possible to efficiently compute $c = c_1, \ldots, c_n$ and $d = d_1, \ldots, d_n$ given $a, b$ and $C_r$.) The predicate $Q$ is then defined by
\begin{gather*}
Q(t, \pi, a, b, r, u) = R_r(t, \pi, u \oplus c_{i_1} \cdots c_{i_k})\;.	
\end{gather*}
\end{definition}
\end{enumerate}
\end{protocol}

\subsection{Replacing Clifford verification with XZ verification in Protocol \ref{protocol:bjsw16}}
\label{section:modifiedqma}

The authors of \cite{qma} introduce a zero-knowledge proof system which allows the verifier to determine whether the prover holds a state that has sufficiently low energy with respect to a $k$-local Clifford Hamiltonian (see Section 2 of \cite{qma}). In this section, we modify their proof system so that it applies to an input encoded as an instance of the XZ local Hamiltonian problem (Definition~\ref{def:2-local-xz-problem}) rather than as an instance of the Clifford Hamiltonian problem.

Before we introduce our modifications, we explain why it is necessary in the first place to alter the proof system presented in \cite{qma}. Modulo the encoding $E$ which the prover applies to its state in Protocol \ref{protocol:bjsw16}, the quantum verifier from the same protocol is required to perform a projective measurement of the form $\{\Pi = C^* \ket{0^k}\bra{0^k} C, \,\Id-\Pi\}$ of the state that the prover sends it (where $C$ is a Clifford unitary acting on $k$ qubits) and reject if it obtains the first of the two possible outcomes. Due to the properties of Clifford unitaries, this action is equivalent to measuring $k$ commuting $k$-qubit Pauli observables $C^* Z_i C$ for $i\in\{1,\ldots,k\}$ (where $Z_i$ is a Pauli $\sigma_Z$ observable acting on the $i$th qubit), and rejecting if all of said measurements result in the outcome $+1$. 

Our goal is to replace the quantum component of the verifier's actions in Protocol \ref{protocol:bjsw16}—a component which, fortunately, consists entirely of performing the projective measurement just described—with the measurement protocol introduced in~\cite{measurement} (summarized as Protocol \ref{protocol:mah18}). Unfortunately, the latter protocol 1. only allows for standard and Hadamard basis measurements, and 2. does not accommodate a verifier who wishes to perform multiple successive measurements on the same qubit: for each qubit that the verifier wants to measure, it must decide on a measurement basis (standard or Hadamard) prior to the execution of the protocol, and once made its choices are fixed for the duration of its interaction with the prover. This allows the verifier to, for example, obtain the outcome of a measurement of the observable $C^* Z_i C$ for some \emph{particular} $i$, by requesting measurement outcomes of all $k$ qubits in the appropriate basis and taking the product of the outcomes obtained. However, it is not obvious how the same verifier could request the outcome of measuring a $k$-\emph{tuple} of commuting Pauli observables which all act on the same $k$ qubits. 

To circumvent this technical issue, we replace the Clifford Hamiltonian problem used in~\cite{qma} with the QMA-complete XZ Hamiltonian problem. The advantage of this modification is that it becomes straightforward to implement the required energy measurements using the measurement protocol from \cite{measurement}. In order to make the change, we require that the verifier's measurements act on a linear, rather than a constant, number of qubits with respect to the size of the problem input.

A different potentially viable modification to the proof system of \cite{qma} is as follows. Instead of replacing Clifford Hamiltonian verification with XZ Hamiltonian verification, we could also repeat the original Clifford-Hamiltonian-based protocol a polynomial number of times. In such a scheme, the honest prover would hold $m$ copies of the witness state (as it does in Protocol \ref{protocol:mf16}). The verifier, meanwhile, would firstly choose a random term $C_r^* \ket{0^k} \bra{0^k} C_r$ from the Clifford Hamiltonian, and then select $m$ random Pauli observables of the form $C_r^* Z_i C_r$---where $C_r$ is the particular $C_r$ which it picked---to measure. (For each repetition, $i$ would be chosen independently and uniformly at random from the set $\{1, \ldots, k\}$.) The verifier would accept if and only if the number of times it obtains $-1$ from said Pauli measurements is at least $\frac{m}{2k}$. This approach is very similar to the approach we take for XZ Hamiltonians (which we explain below), and in particular also fails to preserve the perfect completeness of the original protocol in \cite{qma}. For simplicity, we choose the XZ approach. We now introduce the alterations which are necessary in order to make it viable.

Firstly, we require that the honest prover possesses polynomially many copies of the witness state $\sigma$, instead of one. We do this because we want the honest verifier to accept the honest prover with probability exponentially close to 1, which is not naturally true in the verification procedure for 2-local XZ Hamiltonians presented by Morimae and Fitzsimons in \cite{xz}, but which is true in our amplified variant, Protocol \ref{protocol:mf16}. Secondly, we need to modify the verifier's conditions for acceptance. In \cite{qma}, as we have mentioned, these conditions are represented by a predicate $Q$ (that in turn evaluates a predicate $R_r$; see Definitions \ref{def:original-R_r} and \ref{def:original-Q}).

We now describe our alternative proof system for QMA, and claim that it is zero-knowledge. Because the protocol is very similar to the protocol from \cite{qma}, this can be seen by following the proof of zero-knowledge in \cite{qma}, and noting where our deviations require modifications to the reasoning. On the other hand, we do not argue that the proof system is complete and sound, as we do not need to make explicit use of these properties. (Intuitively, however, the completeness and the soundness of the proof system follow from those of Protocol \ref{protocol:mf16}, and the soundness of the latter is a property which we will use.)

\begin{protocol}[Alternative proof system for QMA]
\label{protocol:modifiedqma}
\quad \\[-0.5cm]

\emph{Notation.}
Refer to notation section of Protocol \ref{protocol:mf16}.

\emph{Parties.}
The proof system involves 
\begin{enumerate}[itemsep=1pt, topsep=0pt]
\item A \emph{verifier}, who implements a quantum polynomial-time procedure;
\item A \emph{prover}, who is unbounded, but who is only required by the protocol to implement a quantum polynomial-time procedure.
\end{enumerate}
The verifier and the prover communicate quantumly.

\emph{Inputs.} 
\begin{enumerate}[itemsep=1pt, topsep=0pt]
\item Input to the verifier:
\begin{enumerate}[itemsep=1pt, topsep=0pt]
\item The Hamiltonian $H$, and the numbers $a$ and $b$.
\item A quantum computationally concealing, perfectly binding (classical) commitment protocol.
\item A proof system for NP sound against arbitrary quantum provers.
\end{enumerate}
\item Input to the prover:
\begin{enumerate}[itemsep=1pt, topsep=0pt]
\item The Hamiltonian $H$, and the numbers $a$ and $b$.
\item The $n$-qubit quantum state $\rho = \sigma^{\otimes n}$, where $\sigma$ is the ground state of the Hamiltonian $H$.
\item A quantum computationally concealing, perfectly binding (classical) commitment protocol.
\item A zero-knowledge proof system for NP sound against arbitrary quantum provers.
\end{enumerate}
\end{enumerate}

\emph{Protocol.} 
\begin{enumerate}[itemsep=1pt,topsep=0pt]
    \item \textit{Prover's encoding step:} The same as the prover's encoding step in Protocol \ref{protocol:bjsw16}, except that $t \in \{0, +\}^N$ rather than $\{0, +, +_y\}^N$. (This change will be justified in the proof of Lemma~\ref{lem:modified-qma-proof}.)
    \item \textit{Coin flipping protocol:} Unmodified from Protocol \ref{protocol:bjsw16}, except that $r=(r_1,\ldots,r_m)$ represents the choice of $m$ terms from the 2-local XZ Hamiltonian $H$ (with the choices being made as described in step 2 of Protocol \ref{protocol:mf16}) instead of a random term from a Clifford Hamiltonian. Note that $r$ determines the indices of the $2m$ logical qubits which the verifier will measure in step 3.
    \item \textit{Verifier's challenge:} The same as the verifier's challenge in Protocol \ref{protocol:bjsw16}, except that the verifier now applies $U_r$ transversally instead of $C_r$. (See item 2(c) in Definition \ref{def:R_r} below for the definition of $U_r$.)
    \item \textit{Prover's response:} The same as Protocol \ref{protocol:bjsw16} (but note that the predicate $Q$, which the prover checks and then proves is satisfied, is the $Q$ described in Definition \ref{def:Q} below).
\end{enumerate}
\end{protocol}

\begin{definition}[Redefinition of $R_r$]
\label{def:R_r}
Let $i_1, \ldots, i_{2m}$ be the indices of the logical qubits which were chosen for measurement in step 2 of Protocol \ref{protocol:modifiedqma}, ordered by their corresponding $j$s (so that $i_1$ and $i_2$ are the qubits that were measured in order to determine whether $H_{s_1}$ was satisfied, and so on). Let $u_{i_1}, \ldots, u_{i_{2m}}$ be the $2N$-bit strings which the verifier claims are the classical states that remained after said measurements were performed, and for each $i \in \{i_1, \ldots, i_{2m}\}$, define $N$-bit strings $q_i, z_i$ such that $\pi(q_i || z_i) = u_i$ (alternatively: $\pi^{-1}(u_i) = q_i || z_i$). In Protocol \ref{protocol:modifiedqma}, the predicate $R_r(t, \pi, u)$ takes the value 1 if and only if the following conditions are met:
\begin{enumerate}[itemsep=1pt,topsep=0pt]
\item $q_i \in \mathcal{D}_N$ for every $i \in \{i_1, \ldots, i_{2m}\}$.
\item The number $\frac{\textsc{Count}}{m}$ (where \textsc{Count} is obtained by executing the following procedure) is closer to $\frac{1}{2} - \frac{a}{\sum_s 2 |d_s|}$ than to $\frac{1}{2} - \frac{b}{\sum_s 2 |d_s|}$.
\begin{enumerate}[itemsep=1pt, topsep=0pt]
\item Initialise \textsc{Count} to 0.
\item For each $j \in \{1, \ldots, m\}$: Suppose that $H_{s_j} = d_j P_1 P_2$, for some $P_1,P_2\in\{\sigma_X,\sigma_Z\}$. The tuple $(P_1, u_{2j - 1}, P_2, u_{2j})$ determines a `logical' measurement result that could equally have been obtained by measuring $H_{r_j} \sigma$, where $\sigma$ is the unencoded witness state. We denote this measurement result by $\lambda$. If $\lambda = -\mathrm{sign}(d_j)$, add one to \textsc{Count}.
\item Let $U_r$ be the circuit obtained from the following procedure:
\begin{enumerate}[itemsep=1pt, topsep=0pt]
\item For each $j \in \{1, \ldots, m\}$, replace any $\sigma_X$s in the term $H_{s_j}$ with $H$ (Hadamard) gates, and replace any $\sigma_Z$s in $H_{s_j}$ with $I$. (For example, if $H_{r_j} = \pi_j \sigma_{X,{\ell_1}} \sigma_{Z,{\ell_2}}$, where the second subscript denotes the index of the qubit on which the observable in question acts, then $U_j = H_{\ell_1} I_{\ell_2}$, where the subscripts $\ell_1$ and $\ell_2$ once again the denote the indices of the qubits on which the gates $H$ and $I$ act.)
\item Apply $U_j$ to the qubits indexed $(j-1)n + 1$ through $jn$.
\end{enumerate}
It must then be the case that $\bra{ z_{i_1} \cdots z_{i_{2m}} } U_r^{\otimes N} \ket{ t_{i_1} \cdots t_{i_{2m}} } \neq 0$ (where each $t_i$ is an $N$-bit string that represents the pattern of trap qubits which was concatenated to the $i$th logical qubit during step 1 of Protocol \ref{protocol:modifiedqma}).
\end{enumerate}
\end{enumerate}
\end{definition}

\begin{definition}[Redefinition of $Q$]
\label{def:Q}
Let $c_1, \ldots, c_n, \: d_1, \ldots, d_n \in \{0,1\}^{2N}$ be the unique strings such that
\begin{gather*}
U_r^{\otimes 2N}(X^{a_1} Z^{b_1} \otimes \cdots \otimes X^{a_n} Z^{b_n}) = \alpha(X^{c_1} Z^{d_1} \otimes \cdots \otimes X^{c_n} Z^{d_n}) U_r^{\otimes 2N}
\end{gather*}
for some $\alpha \in \{1, i, -1, -i\}$. (It is possible to efficiently compute $c = c_1, \ldots, c_n$ and $d = d_1, \ldots, d_n$ given $a, b$ and $U_r$. In particular, recalling that $U_r$ is a tensor product of $H$ and $I$ gates, we have that $c_i = a_i$ and $d_i = b_i$ for all $i$ such that the $i$th gate in $U_r^{\otimes 2N}$ is $I$, and $c_i = b_i$, $d_i = a_i$ for all $i$ such that the $i$th gate in $U_r^{\otimes 2N}$ is $H$.) The predicate $Q$ is then defined by
\begin{gather*}
Q(t, \pi, a, b, r, u) = R_r(t, \pi, u \oplus c_{i_1} \cdots c_{i_k})\;,
\end{gather*}
where $R_r$ is as in Definition \ref{def:R_r}.
\end{definition}

\begin{lemma}
\label{lem:modified-qma-proof}
The modified proof system for QMA in Protocol \ref{protocol:modifiedqma} is computationally zero-knowledge for quantum polynomial-time verifiers.
\end{lemma}

\begin{proof} We follow the argument from~\cite[Section 5]{qma}. Steps 1 to 3 only make use of the security of the coin-flipping protocol, the security of the commitment scheme, and the zero-knowledge properties of the NP proof system, none of which we have modified. Step 4 replaces the real witness state $\rho$ with a simulated witness $\rho_r$ that is guaranteed to pass the challenge indexed by $r$; this we can do also (see Remark \ref{remark:rho_r}). Step 5 uses the Pauli one-time-pad to twirl the cheating verifier, presuming that the honest verifier would have applied a Clifford term indexed by $r$ before measuring. We note that, since $U_r$ is a Clifford, the same reasoning applies to our modified proof system.

Finally, using the fact that the Pauli twirl of step 5 restricts the cheating verifier to XOR attacks, step 6 from \cite[Section 5]{qma} proves the following statement: if the difference $|p_0 - p_1|$ is negligible (where $p_0$ and $p_1$ are the probabilities that $\rho$ and $\rho_r$ respectively pass the verifier's test in an honest prover-verifier interaction indexed by $r$), then the channels $\Psi_0$ and $\Psi_1$ implemented by the cheating verifier in each case are also quantum computationally indistinguishable. It follows from this statement that the protocol is zero-knowledge, since, in an honest verifier-prover interaction indexed by $r$, $\rho_r$ would pass with probability 1, and $\rho$ would pass with probability $1 - \textsf{negl}(N)$. (This latter statement is true both in their original and in our modified protocol.) The argument presented in \cite{qma} considers two exclusive cases: the case when $|v|_1 < K$, where $v$ is the string that the cheating verifier XORs to the measurement results, $|v|_1$ is the Hamming weight of that string, and $K$ is the minimum Hamming weight of a nonzero codeword in $\mathcal{D}_N$; and the case when $|v|_1 \geq K$. The analysis in the former case translates to Protocol \ref{protocol:modifiedqma} without modification, but in the latter case it needs slight adjustment.

In order to address the case when $|v|_1 \geq K$, Broadbent et al.\ use a lemma which---informally---states that the action of a Clifford on $k$ qubits, each of which is initialised uniformly at random to one of $\ket{0}, \ket{+}$, or $\ket{+}_y$, has at least a $3^{-k}$ chance of leaving at least one out of $k$ qubits in a standard basis state. We may hesitate to replicate their reasoning directly, because our $k$ (the number of qubits on which our Hamiltonian acts) is not a constant. While it is possible that a mild modification suffices to overcome this problem, we note that in our case there is a simpler argument for an analogous conclusion: since $U_r$ is a tensor product of only $H$ gates and $I$ gates, it is straightforward to see that, if each of the $2m$ qubits on which it acts is initialised either to $\ket{0}$ or to $\ket{+}$, then 1) each of the $2m$ qubits has exactly a 50\% chance of being left in a standard basis state, and 2) the states of these $2m$ qubits are independent.

Now we consider the situation where a string $v = v_1 \: v_2 \: \cdots \: v_{2m}$, of length $4mN$ and of Hamming weight at least $K$, is permuted (`permuted', here, means that $\pi \in S_{2N}$ is applied to each $v_i$ individually) and then XORed to the result of measuring $4mN$ qubits ($2m$ blocks of $2N$ qubits each) in the standard basis after $U_r$ has been transversally applied to those qubits. It is straightforward to see, by an application of the pigeonhole principle, that there must be at least one $v_i$ whose Hamming weight is $\geq \frac{K}{2m}$. Consider the result of XORing this $v_i$ to its corresponding block of measured qubits. Half of the $2N$ qubits in that block would originally have been encoding qubits, and half would have been trap qubits; half again of the latter, then, would have been trap qubits left in a standard basis state by the transversal action of $U_r$. As such, the probability that none of the 1-bits of $v_i$ are permuted into positions which are occupied by the latter kind of qubit is $(\frac{3}{4})^{-\frac{K}{2m}}$, which is negligibly small as long as $K$ is made to be a higher-order polynomial in $N$ than $2m$ is. The remainder of the argument in \cite[Section 5]{qma} follows directly.
\end{proof}

\section{The protocol}
\label{section:protocol}

In this section, we present our construction of a zero-knowledge argument system for QMA. Our argument system allows a classical probabilistic polynomial-time verifier and a quantum polynomial-time prover to verify that any problem instance $x$ belongs to any particular language $L \in$ QMA, provided that the prover has access to polynomially many copies of a valid quantum witness for an instance of the $2$-local XZ local Hamiltonian problem to which $x$ is mapped by the reduction implicit in Theorem~\ref{thm:xz-complete}. The argument system is sound (against quantum polynomial-time provers) under the following assumptions:

\begin{ass}\label{assumptions}~
\begin{enumerate}[itemsep=1pt, topsep=0pt]
\item The Learning With Errors problem (LWE)~\cite{regev2009lattices} is quantum computationally intractable. (Specifically, we make the same asssumption about the hardness of LWE that is made in~\cite[Section 9]{measurement} in order to prove the soundness of the measurement protocol.)
\item There exists a commitment scheme $(\textsf{gen}, \textsf{initiate}, \textsf{commit}, \textsf{reveal}, \textsf{verify})$ of the form described in Appendix~\ref{appendix:commitment} that is unconditionally binding and quantum computationally concealing. (This assumption is necessary to the soundness of the proof system presented in~\cite{qma}.) It is known that a commitment scheme with the properties required can be constructed assuming the quantum computational hardness of LWE~\cite{coladangelo2019non}, although the parameters may be somewhat different from those required for soundness.
\end{enumerate}
\end{ass}

The following exposition of our protocol relies on definitions from Section~\ref{section:ingredients}, and we encourage the reader to read that section prior to approaching this one. We also direct the reader to Figures~\ref{figure:honest-protocol} and \ref{figure:cheating-verifier} for diagrams that chart the protocol's structure.

\begin{protocol}
\label{protocol:main-protocol}
\emph{Zero-knowledge, classical-verifier argument system for QMA.}

\emph{Notation.}
Let $L$ be any promise problem in QMA, and let $(H = \sum_{s=1}^S d_s H_s, a, b)$ be an instance of the 2-local XZ Hamiltonian problem to which $L$ can be reduced (see Definition \ref{def:2-local-xz-problem} and Theorem~\ref{thm:xz-complete}).
Define
\begin{gather*}
\pi_s \,=\, \frac{|d_s|}{\sum_s |d_s|}\;.
\end{gather*} 
Following \cite{qma}, we take the security parameter for this protocol to be $N$, the number of qubits in which the concatenated Steane code used during the encoding step of the protocol (step 1) encodes each logical qubit. We assume, accordingly, that $N$ is polynomial in the size of the problem instance $x$. 

\emph{Parties.}

The protocol involves
\begin{enumerate}[itemsep=1pt, topsep=0pt]
\item A \emph{verifier}, which runs in classical probabilistic polynomial time;
\item A \emph{prover}, which runs in quantum polynomial time.
\end{enumerate}

\emph{Inputs.}
The protocol requires the following primitives: 
\begin{itemize}[itemsep=1pt, topsep=0pt]
\item A perfectly binding, quantum computationally concealing commitment protocol (\textsf{gen}, \textsf{initiate}, $\mathsf{commit},\mathsf{reveal}$, \textsf{verify}) (which will be used twice: once for the prover's commitment in step 2, and then again for the coin-flipping protocol in step 3). We assume that this commitment protocol is of the form described in Appendix \ref{appendix:commitment}.
\item A zero-knowledge proof system for NP.
\item An extended trapdoor claw-free function family (ETCFF family), as defined in \cite{measurement}. (Note that we fall short of using the ETCFF family as a black box: for the trapdoor check of step 8, we rely on the specific properties of the LWE-based construction of an ETCFF family that \cite{measurement} provides. See Appendix \ref{appendix:trapdoor-check} for details.) 
\end{itemize}
Apart from the above cryptographic primitives, we assume that the verifier and the prover also receive the following inputs.
\begin{enumerate}[noitemsep,topsep=0pt]
\item Input to the verifier: the Hamiltonian $H$ and the numbers $a$ and $b$.
\item Input to the prover: the Hamiltonian $H$, the numbers $a$ and $b$, and the quantum state $\rho = \sigma^{\otimes m}$, where $\sigma$ is a ground state of the Hamiltonian $H$. 
\end{enumerate}

\emph{Protocol.} 
\begin{enumerate}[itemsep=1pt, topsep=0pt]
\item \emph{The prover encodes the witness.}
The prover encodes the quantum witness \(\rho\) by applying the following steps:
\begin{enumerate}[itemsep=1pt, topsep=0pt]
\item Apply concatenated Steane code
\item Concatenate trap qubits \(|t\rangle\)
\item Apply a random permutation \(\pi\)
\item Apply a Pauli one-time-pad \(X^a Z^b\)
\end{enumerate}
The encoding process here is the same as that from step 1 of Protocol \ref{protocol:modifiedqma}; we direct the reader to Protocol \ref{protocol:modifiedqma}, and the Protocol \ref{protocol:bjsw16} to which it refers, for a more detailed explanation of the steps. Denote the application of the prover's encoding to the state $\rho$ by $E(\rho)$. 
\item \emph{The prover commits to its encoding keys.} The prover commits to the strings \((\pi, a, b)\) from the previous step, using randomness \(s_p\). Call the prover's commitment string \(z\), so that \(z = \mathsf{commit}((\pi, a, b), s_p) \).
\item \emph{The verifier and the prover execute the first half of a two-stage coin-flipping protocol.}\footnote{We need to execute the coin-flipping protocol in two stages because, in our (classical-verifier) protocol, the prover cannot physically send the quantum state $E(\rho)$ to its verifier before the random string $r$ is decided, as the prover of Protocol \ref{protocol:bjsw16} does. If we allow our prover to know $r$ at the time when it performs measurements on the witness $\rho$, it will trivially be able to cheat.} The verifier commits to \(r_v\), its part of the random string that will be used to determine which random terms in the Hamiltonian $H$ it will check in subsequent stages of the protocol. Let \(c = \mathsf{commit}(r_v, s_v)\). The prover sends the verifier \(r_p\), which is its own part of the random string. The random terms will be determined by \(r = r_v \oplus r_p\). ($r$ is used to determine these terms in the same way that $r$ is used in Protocol \ref{protocol:mf16}.)
\item \emph{The verifier initiates the measurement protocol. (Refer to Protocol \ref{protocol:mah18} for an outline of the steps in said measurement protocol.)} The verifier chooses the measurement bases $h = h_1 \cdots h_{2nN}$ in which it wishes to measure the state $E(\rho)$. $2kN$ out of the $2nN$ bits of $h$---corresponding to $k$ logical qubits---are chosen so that the verifier can determine whether $\sigma$ satisfies the Hamiltonian terms specified by $r = r_v \oplus r_p$. In our particular case, $k = 2m$, where $m$ is the number of Hamiltonian terms that the verifier will check are satisfied. For the remaining qubits $i$, the verifier sets $h_i$ to 0. The verifier sends the function keys $\kappa = \kappa_1, \ldots, \kappa_{2nN}$ to the prover.
\item \emph{The prover commits to its encoded witness state, as per the measurement protocol.} The prover commits to the quantum state \(E(\rho)\) by concatenating a preimage register to each qubit in $E(\rho)$, applying the functions specified by $\kappa_1, \ldots, \kappa_{2nN}$ in superposition as Protocol \ref{protocol:mah18} describes, measuring the resulting output superpositions, and sending the outcomes \(y_1, \ldots, y_{2nN}\) to the verifier.
\item	\emph{The verifier chooses at random to run either a test round or a Hadamard round.} Each option is chosen with probability $\frac{1}{2}$.
\begin{enumerate}[itemsep=1pt,topsep=0pt]
	\item If a test round is chosen, the verifier requests standard basis measurements of each committed qubit $\ket{\psi_i}$ in $E(\rho)$ and its associated preimage register, and receives measurement results $\beta_i, x_i$ for $i \in \{1, \ldots, 2nN\}$ from the prover. It then checks, for each $i \in \{1, \ldots, 2nN\}$, whether the $i$th measurement result is a preimage of $y_i$. (See Definition \ref{def:preimage-of-y} for a definition of the phrase `preimage of $y$'.) More precisely: if $h_i = 0$, the verifier accepts iff $y_i$ is in the support of $g_{\kappa_i}(\beta_i \| x_i)$ (recalling that $g_{\kappa_i}(\beta_i \| x_i)$ is a probability distribution); and if $h_i = 1$, the verifier accepts iff $y_i$ is in the support of $f_{\kappa_i}(\beta_i \| x_i)$.
	\item If a Hadamard round is chosen, the verifier requests Hadamard basis measurements of each committed qubit and its associated preimage register, and receives $\beta_i, d_i$ for $i \in \{1, \ldots, 2nN\}$ from the prover. It then obtains the measurement results it desires using the procedure outlined in Section 2.2 of \cite{measurement}:
	\begin{enumerate}[itemsep=1pt,topsep=0pt]
	\item If $h_i = 0$, the verifier ignores the measurement results $\beta_i, d_i$, and uses the trapdoor $\tau_{\kappa_i}$ to invert $y_i$. The inversion procedure outputs a preimage $(m_i, x_{m_i, y_i})$; the verifier stores $m_i$ as the standard-basis measurement outcome. (If the inverse does not exist, the verifier stores a random bit as $m_i$ and rejects.)
	\item If $h_i = 1$, the verifier computes the two preimages $x_{0, i}$ and $x_{1, i}$ to $y_i$ under $f_{\kappa_i}$ using the trapdoor $\tau_{\kappa_i}$. It then stores $\beta_i \oplus d_i \cdot (x_{0,i} \oplus x_{1,i})$ as the Hadamard basis measurement outcome. (Once again, if either inverse does not exist, the verifier stores a random bit as $m_i$ and rejects.)
	\end{enumerate}
\end{enumerate}
\item \emph{The verifier opens its commitment, and sends the prover its measurement outcomes and function trapdoors.} If a test round was chosen in the previous step, the verifier now accepts the prover directly, provided that the preimage test was satisfied for all $i$. (If not all the preimage tests passed, the verifier directly    
rejects the prover.) If a Hadamard round was chosen, and the verifier instead opens its commitment to \(r_v\) and sends its measurement outcomes \(m = m_1, \ldots, m_{2nN}\) to the prover, along with its trapdoors $\tau_{\kappa_1}, \ldots, \tau_{\kappa_{2nN}}$.
\item \emph{The prover checks that the verifier's measurement outcomes and trapdoors were consistent; it then proves that the verifier's measurement outcomes were consistent.} The prover checks that $\tau_{\kappa_1}, \ldots, \tau_{\kappa_{2nN}}$ are indeed trapdoors for the functions specified by $\kappa_1, \ldots, \kappa_{2nN}$, and also that $\kappa_1, \ldots, \kappa_{2nN}$ are valid ETCFF keys, using the procedure described in Protocol \ref{protocol:trapdoor-and-key-check}. 
It also defines $u = u_{i_1} \cdots u_{i_{2m}}$ (each $u_i$ is $2N$ bits long) $= m_{\ell_1} \cdots m_{\ell_{4mN}}$, where $\ell_1, \ldots, \ell_{4mN}$ are the indices of the qubits on which $U_r^{\otimes 2N}$ acts nontrivially, and checks that $u$ causes the predicate $Q(t, \pi, a, b, r, u)$ to be satisfied. (The $Q$ we refer to here is the $Q$ of Definition \ref{def:Q}. We define $U_r$ in the same way that $U_r$ was defined in Definition \ref{def:R_r}.) If either of these tests fails, the prover aborts. If both tests pass, then the prover proves, using an NP zero-knowledge proof system,\footnote{It was shown in~\cite{watrous2009zero} that the second item in Assumptions~\ref{assumptions} suffices to guarantee the existence of a proof system for languages in NP that is zero-knowledge against quantum polynomial-time verifiers. Our proof that our protocol is zero-knowledge for \emph{classical} verifiers only requires that the NP proof system used here is (likewise) zero-knowledge against classical verifiers; however, it becomes necessary to require post-quantum security of this proof system if we want our protocol also to be zero-knowledge for potentially quantum malicious verifiers.} that the verifier's outcomes are consistent in the following sense:
\begin{enumerate}[itemsep=1pt, topsep=0pt]
\item[] The verifier's outcomes $u$ are consistent if there exists a string \(s_p\) and an encoding key \((t, \pi, a, b)\) such that \(z = \mathsf{commit}((\pi, a, b), s_p) \) and \(Q(t, \pi, a, b, r, u) = 1\).
\end{enumerate}
\end{enumerate}
\end{protocol}

\begin{figure}[H]
\begin{center}
\includegraphics[width=0.8\textwidth]{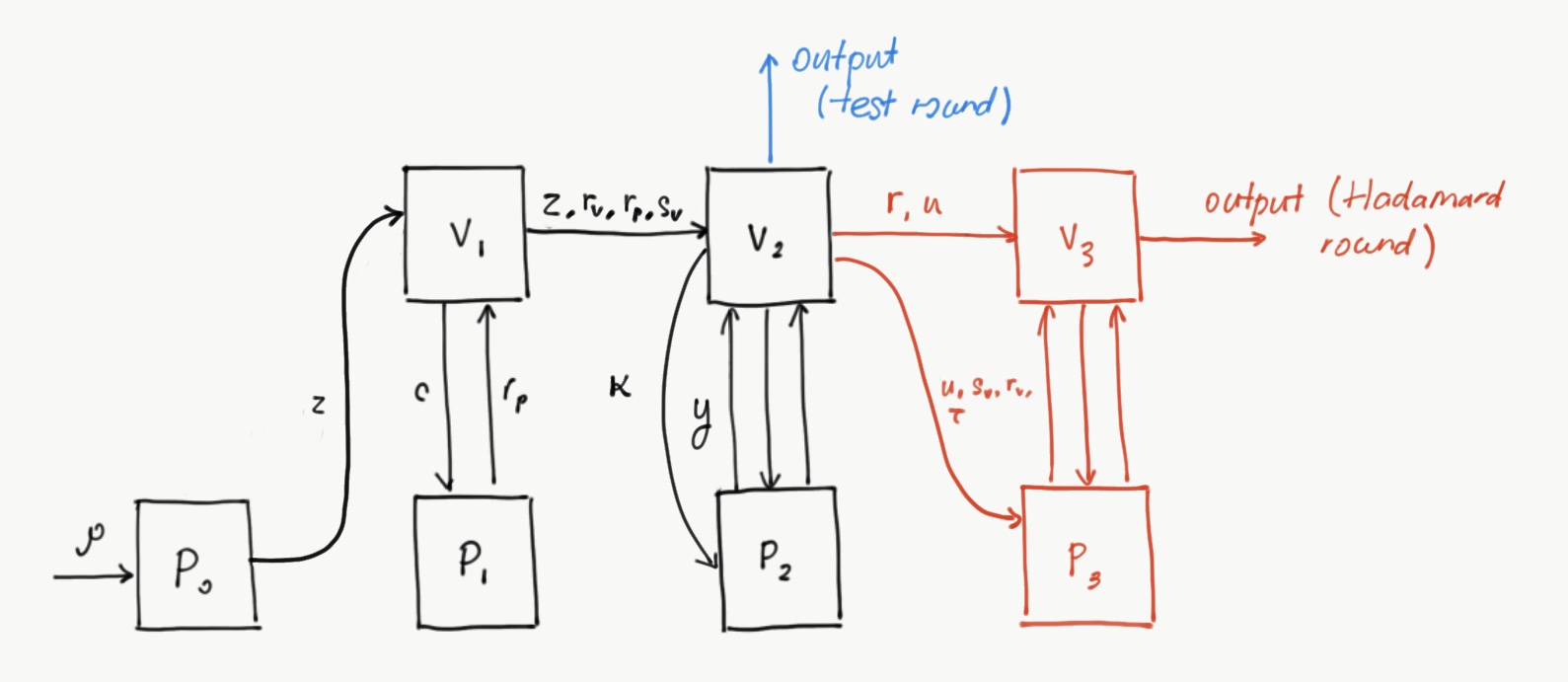} \\
\caption{
\label{figure:honest-protocol}
Diagrammatic representation of an honest execution of Protocol \ref{protocol:main-protocol}. We omit communication between the different parts of the prover for neatness, and we also omit the initial messages $i$ (see Appendix \ref{appendix:commitment}) from executions of the perfectly binding, quantum computationally concealing commitment protocol which we refer to in Assumptions \ref{assumptions}. The blue parts of the diagram indicate what occurs in the case of a test round, and the red parts indicate what occurs in the case of a Hadamard round.}
\end{center}
\end{figure}
\begin{figure}[H]
\begin{center}
\includegraphics[width=0.8\textwidth]{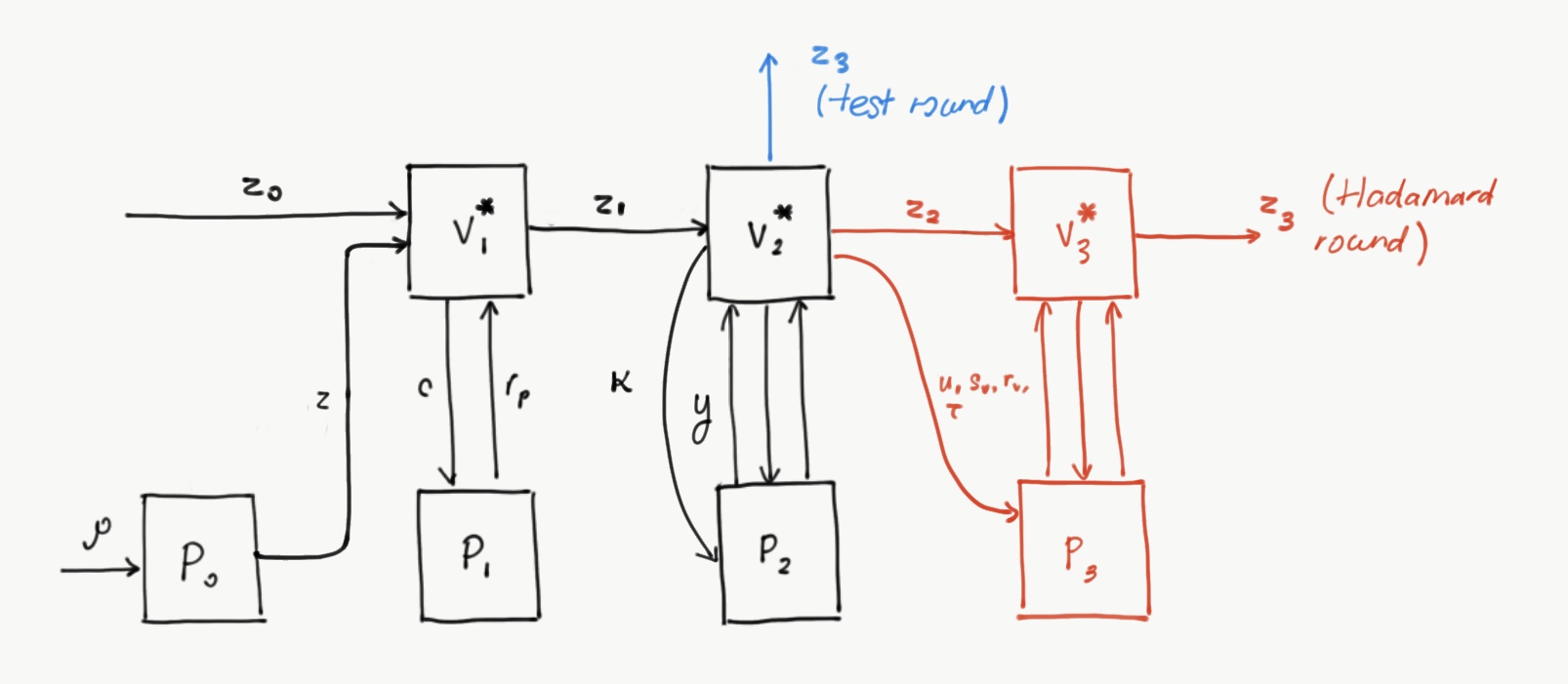}
\end{center}
\caption{
\label{figure:cheating-verifier}
Diagrammatic representation of Protocol \ref{protocol:main-protocol} with a cheating verifier. The cheating verifier $V^*$ may take some (classical) auxiliary input $Z_0$, store auxiliary information (represented by $Z_1$ and $Z_2$), and produce a final output $Z_3$ that deviates from that specified by the protocol.}
\end{figure}

\section{Completeness of protocol}

\begin{lemma}
\label{lemma:completeness}
Suppose that the instance $x = (H, a, b)$ of the 2-local XZ Hamiltonian problem that is provided as input to the verifier and prover in Protocol \ref{protocol:main-protocol} is a yes-instance, i.e. the ground energy of $H$ is smaller than $a$. Then, the probability that the honest verifier accepts after an interaction with the honest prover in Protocol \ref{protocol:main-protocol} is $1 - \mu(|x|)$, for some negligible function $\mu$.
\end{lemma}

\begin{proof}
The measurement protocol outlined in section \ref{section:mah18} has the properties that

\begin{enumerate}
\item for any $n$-qubit quantum state $\rho$ and for any choice of measurement bases $h$, the honest prover is accepted by the honest verifier with probability $1 - \mathsf{negl}(n)$, and
\item the distribution of measurement outcomes obtained by the verifier from an execution of the measurement protocol (the measurement outcomes $m_i$ in step 6(b) of Protocol \ref{protocol:main-protocol}) is negligibly close in total variation distance to the distribution that would have been obtained by performing the appropriate measurements directly on $\rho$.
\end{enumerate}
These properties are stated in Claim 5.3 of \cite{measurement}. It is evident (assuming the NP zero-knowledge proof system has perfect completeness) that if the verifier of Protocol \ref{protocol:main-protocol} had obtained the outcomes $m$ through direct measurement of $\rho$, it would accept with exactly the same probability with which the verifier of Protocol \ref{protocol:mf16} would accept $\rho = \sigma^{\otimes n}$. By Claim \ref{claim:proof-amplification-works}, this latter probability is exponentially close to 1. Lemma \ref{lemma:completeness} follows.
\end{proof}

\section{Soundness of protocol}
\label{section:soundness}
Let the honest verifier of the argument system in Protocol \ref{protocol:main-protocol} be denoted $V$, and let an arbitrary quantum polynomial-time prover with which $V$ interacts be denoted $\mathbb{P}$. For this section, we will require notation from Section 5.3 of \cite{measurement}, the proof of Theorem 8.6 of the same paper, and Section 4 of \cite{qma}. We will by and large introduce this notation as we proceed (and some of it has been introduced already in Sections \ref{section:mah18} and \ref{section:bjsw16}, the sections containing outlines of the measurement protocol from \cite{measurement} and the zero-knowledge proof system from \cite{qma}), but the reader should refer to the above works if clarification is necessary.

We begin by making some preliminary definitions and proving a claim, from which the soundness of Protocol \ref{protocol:main-protocol} (Lemma \ref{lem:soundness-of-main-protocol}) will naturally follow. Firstly, we introduce some notation from Section 4 of \cite{qma}:

\begin{definition}[Projection operators $\Pi_0$ and $\Pi_1$]
Define $N$ as it is defined in Protocol \ref{protocol:main-protocol}. Let $\mathcal{D}^0_N$ be the set of bitstrings $x$ such that the encoding of $\ket{0}$ under the concatenated Steane code of Protocol \ref{protocol:bjsw16} (or of Protocol \ref{protocol:main-protocol}) is $\sum_{x \in \mathcal{D}^0_N} \ket{x}$, and let $\mathcal{D}^1_N$ likewise be the set of bitstrings $x$ such that the encoding of $\ket{1}$ under the concatenated Steane code is $\sum_{x \in \mathcal{D}^1_N} \ket{x}$. (See Definition \ref{def:original-R_r}, and Section A.6 of \cite{qma}, for details about the concatenated Steane code. The first condition in Definition \ref{def:original-R_r} will provide some motivation for the following definitions of $\Pi_0$ and $\Pi_1$.) Define
\begin{align*}
&\Pi_0 = \sum_{x \in \mathcal{D}^0_N} \ket{x}\bra{x}\;,
&\Pi_1 = \sum_{x \in \mathcal{D}^1_N} \ket{x}\bra{x}\;.
\end{align*}
\end{definition}

\begin{definition}[Projection operators $\Delta_0$ and $\Delta_1$]
Define $N$ as it is defined in Protocol \ref{protocol:main-protocol}. Let $\Delta_0$ and $\Delta_1$ be the following projection operators:
\begin{align*}
&\Delta_0 = \frac{I^{\otimes N} + Z^{\otimes N}}{2}\;,
&\Delta_1 = \frac{I^{\otimes N} - Z^{\otimes N}}{2}\;.
\end{align*}
$\Delta_0$ is the projection onto the space spanned by all even-parity computational basis states, and $\Delta_1$ is its equivalent for odd-parity basis states. Note that, since all the codewords in $\mathcal{D}_0$ have even parity, and all the codewords in $\mathcal{D}_1$ have odd parity, it holds that $\Pi_0 \leq \Delta_0$ and that $\Pi_1 \leq \Delta_1$.
\end{definition}

\begin{definition}[The quantum channel $\Xi$]
Define a quantum channel mapping $N$ qubits to one qubit as follows:
\begin{gather*}
\Xi_N(\sigma) = \frac{
\langle I^{\otimes N}, \sigma \rangle I +
\langle X^{\otimes N}, \sigma \rangle X +
\langle Y^{\otimes N}, \sigma \rangle Y +
\langle Z^{\otimes N}, \sigma \rangle Z}{2}\;.
\end{gather*}
Loosely, $\Xi_N$ can be thought of as a simplification of the decoding operator to the concatenated Steane code that the honest prover applies to its quantum witness in Protocol \ref{protocol:bjsw16} (or in Protocol \ref{protocol:main-protocol}). Its adjoint is specified by
\begin{gather*}
\Xi^*_N(\sigma) = \frac{
\langle I, \sigma \rangle I^{\otimes N} +
\langle X, \sigma \rangle X^{\otimes N} +
\langle Y, \sigma \rangle Y^{\otimes N} +
\langle Z, \sigma \rangle Z^{\otimes N}}{2}\;,
\end{gather*}
and has the property that
\begin{align*}
&\Xi^*_N(\ket{0}\bra{0}) = \Delta_0\;,
&\Xi^*_N(\ket{1}\bra{1}) = \Delta_1\;,
\end{align*}
a property which we will shortly use.
\end{definition}

Let $z$ be prover $\mathbb{P}$'s commitment string from step 2 of Protocol \ref{protocol:main-protocol}. Because the commitment protocol is perfectly binding, there exists a unique, well-defined tuple $(\pi, a, b)$ and a string $s_p$ such that $z = \textsf{commit}((\pi, a, b), s_p)$.

\begin{definition}
For notational convenience, we define a quantum procedure $M$ on a $2nN$-qubit state $\rho$ as follows:
\begin{enumerate}[itemsep=1pt,topsep=0pt]
\item Apply $X^a Z^b$ to $\rho$, to obtain a state $\rho'$.
\item Apply $\pi^{-1}$ to each $2N$-qubit block in the state $\rho'$, to obtain a state $\rho''$.
\item Discard the last $N$ qubits of each $2N$-qubit block in $\rho''$, to obtain a state $\rho'''$.
\item To each $N$-qubit block in $\rho'''$, apply the map $\Xi_N$.
\end{enumerate}
We also define the procedure $\tilde{M}$ as the application of the first three steps in $M$, again for notational convenience.
\end{definition}

Intuitively, we think of $M$ as an inverse to the prover's encoding procedure $E$. $M$ may not actually invert the prover's encoding procedure, if the prover lied about the encoding key that it used when it sent the verifier $z = \textsf{commit}((\pi, a, b), s_p)$; however, this is immaterial.

We now prove a claim from which the soundness of Protocol \ref{protocol:main-protocol} will follow. Before we do so, however, we make a remark about notation for clarity. When we write `$V$ accepts the distribution $D_{\xi, h}$ with probability $p$' (or similar phrases), we mean that, in \cite{measurement}'s notation from section 8.2,
\[ \sum_{h \in \{0,1\}^{2nN}} v_h (1 - \tilde{p}_h(D_{\xi, h})) = p.\]
Here, $h$ represents the verifier's choice of measurement bases, as before; $v_h$ is the probability that the honest verifier will select the basis choice $h$, and $1 - \tilde{p}_h(D)$ is defined, for any distribution $D$ over measurement outcomes $m \in \{0,1\}^{2nN}$, as the probability that the honest verifier will accept a string drawn from $D$ on basis choice $h$. (When we refer to the latter probability, we assume, following \cite[Section 4]{qma}, that the prover behaves optimally---in terms of maximising the verifier's eventual probability of acceptance---after the verifier sends it measurement outcomes at the end of step 6 in Protocol \ref{protocol:main-protocol}. For the purposes of the present soundness analysis, therefore, we can imagine that the verifier checks the predicate $Q$ itself after step 6, instead of relying on the prover to prove to it during step 8 that $Q$ is satisfied.)


\begin{claim}
\label{claim:if-V-accepts-so-does-V'}
Suppose there exists a quantum state $\xi$ such that the honest verifier $V$ accepts the distribution $D_{\xi, h}$ with probability $p$. Then the state $M(\xi)$ is accepted by the verifier of Protocol \ref{protocol:mf16} with probability at least $p$.
\end{claim}

\begin{proof}
Fix a choice of $r$ (see step 3 of Protocol \ref{protocol:main-protocol} for a definition of $r$). Let $\mathcal{Z}_r$ be the subset of $\{0,1\}^{n}$ such that the verifier of Protocol \ref{protocol:mf16} accepts if and only if the $n$-bit string that results from concatenating the measurement results it obtains in step 4 of said protocol is a member of $\mathcal{Z}_r$. It is unimportant to the analysis what $\mathcal{Z}_r$ actually is; it matters only that it is well-defined.

For this choice of $r$, we can express the probability that the verifier of Protocol \ref{protocol:mf16} accepts a state $\tau$ as
\begin{gather*}
\sum_{z \in \mathcal{Z}_r} \Bigl< U^*_r \ket{z_1, \ldots, z_{n}} \bra{z_1, \ldots, z_{n}} U_r, \tau \Bigr>.
\end{gather*}

(Though only $2m$ of the $n$ qubits in $\tau$ are relevant to $U_r$, we assume here for notational simplicity that $U_r$ is a gate on $n$ qubits, and that the verifier measures all $n$ qubits of $U_r \tau$ and ignores those measurement results which are irrelevant.)

For the same choice of $r$, we can express the probability that the verifier $V$ from Protocol \ref{protocol:main-protocol} will eventually accept the distribution $D_{\xi, h}$ as
\begin{gather*}
p_r = \sum_{z \in \mathcal{Z}_r} \Bigl< (U_r^*)^{\otimes N} (\Pi_{z_1} \otimes \cdots \otimes \Pi_{z_{n}}) (U_r)^{\otimes N}, \tilde{M}(\xi) \Bigr>\;.
\end{gather*}

Following \cite{qma}, we note that
\begin{align*}
&\sum_{z \in \mathcal{Z}_r} \Bigl< (U^*_r)^{\otimes N} (\Pi_{z_1} \otimes \cdots \otimes \Pi_{z_{n}}) (U_r)^{\otimes N}, \tilde{M}(\xi) \Bigr> \\
&\leq
\sum_{z \in \mathcal{Z}_r} \Bigl< (U^*_r)^{\otimes N} (\Delta_{z_1} \otimes \cdots \otimes \Delta_{z_{n}}) (U_r)^{\otimes N}, \tilde{M}(\xi) \Bigr> \\
&=
\sum_{z \in \mathcal{Z}_r} \Bigl< (U^*_r)^{\otimes N} \Big( \Xi^*_N(\ket{z_1}\bra{z_1}) \otimes \cdots \otimes \Xi^*_N(\ket{z_{n}}\bra{z_{n}}) \Big) (U_r)^{\otimes N}, \tilde{M}(\xi) \Bigr> \\
&= \sum_{z \in \mathcal{Z}_r} \Bigl< (\Xi_N^{\otimes n})^* U^*_r \ket{z_1, \ldots, z_{n} } \bra{z_1, \ldots, z_{n}} U_r, \tilde{M}(\xi) \Bigr> \\
&= \sum_{z \in \mathcal{Z}_r} \Bigl< U^*_r \ket{z_1, \ldots, z_{n}} \bra{z_1, \ldots, z_{n}} U_r, M(\xi) \Bigr>\;.
\end{align*}
For the second-to-last equality above, we have used the observation that, for any $n$-qubit Clifford operation $C$, and every $nN$-qubit state $\sigma$,

\begin{equation*}
\Xi_N^{\otimes n} (C^{\otimes N} \sigma (C^{\otimes N})^*) = C \Xi_N^{\otimes n}(\sigma) C^*.
\end{equation*}

This is equation (35) in \cite{qma}, and can be verified directly by considering the definition of $\Xi_{N}$.

We conclude that, if the distribution $D_{\xi, h}$ is accepted by $V$ with probability $p = \sum_r v_r p_r = \sum_{h} v_h (1 - \tilde{p}_h(D_{\xi, h}))$ (where $v_r$ is the probability that a given $r$ will be chosen, and the second expression is simply a formulation in alternative notation of the first), the state $M(\xi)$ is accepted by the verifier of Protocol \ref{protocol:mf16} with probability at least $p$.
\end{proof}

Now we turn to arguing that Protocol \ref{protocol:main-protocol} has a soundness parameter $s$ which is negligibly close to $\frac{3}{4}$.

\begin{lemma}
\label{lem:soundness-of-main-protocol}
Suppose that the instance $x=(H,a,b)$ of the $2$-local XZ Hamiltonian problem that is provided as input to the verifier and prover in Protocol \ref{protocol:main-protocol} is a no-instance, i.e.\ the ground energy of $H$ is larger than $b$. Then, provided that Assumptions~\ref{assumptions} hold, the probability that the honest verifier $V$ accepts in Protocol \ref{protocol:main-protocol} after an interaction with any quantum polynomial-time prover $\mathbb{P}$ is at most $\frac{3}{4} + \textsf{negl}(|x|)$.
\end{lemma}

\begin{proof}
Claim \ref{claim:measurement-protocol} guarantees that, for any arbitrary quantum polynomial-time prover $\mathbb{P}$ who executes the measurement protocol with $V$, there exists a state $\xi$, a prover $\mathbb{P}'$ and a negligible function $\mu$ such that 
\begin{gather}
\|D^C_{\mathbb{P}, h} - D_{\mathbb{P'},h}\|_{TV} \leq \sqrt{p_{h,T}} + p_{h,H} + \mu \;,\quad \text{and} \notag\\
D_{\mathbb{P}',h} \approx_c D_{\xi,h}\;. \label{eq:comp-indist2}
\end{gather}

(See the paragraph immediately above Claim \ref{claim:measurement-protocol} for relevant notation.)

It follows from \eqref{eq:comp-indist2} that, if $V$ accepts the distribution $D_{\mathbb{P}',h}$ with probability $p$, it must accept the distribution $D_{\xi, h}$ with probability $p - \mathsf{negl}(N)$, because the two are computationally indistinguishable and the verifier $V$ is efficient. Therefore (using Claim \ref{claim:if-V-accepts-so-does-V'}), if $V$ accepts $D_{\mathbb{P}',h}$ with probability $p$, the verifier of Protocol 8.3 from \cite{measurement} accepts the state $M(\xi)$ with probability at least $p - \textsf{negl}(N)$. By the soundness of Protocol \ref{protocol:mf16} (Claim \ref{claim:proof-amplification-works}), we conclude that $p = \mathsf{negl}(N)$ when the problem Hamiltonian is a no-instance.

We now apply a similar argument to that which is used in Section 8.2 of \cite{measurement} in order  to establish an upper bound on the probability $\phi$ that $V$ accepts $\mathbb{P}$ in a no-instance. Let $E^H_{\mathbb{P}, h}$ denote the event that the verifier $V$ does not reject the prover labelled $\mathbb{P}$ in a Hadamard round indexed by $h$ \emph{during the measurement protocol phase} of Protocol \ref{protocol:main-protocol}. Let $E^T_{\mathbb{P}, h}$ denote the analogous event in a test round. Furthermore, let $E_{\mathbb{P}, h}$ denote the event that the verifier accepts the prover $\mathbb{P}$ in the last step of Protocol \ref{protocol:main-protocol}. The total probability that $V$ accepts $\mathbb{P}$ is the average, over all possible basis choices $h$, of the probability that $V$ accepts $\mathbb{P}$ after a test round indexed by $h$, plus the probability that $V$ accepts $\mathbb{P}$ after a Hadamard round indexed by $h$. As such, 
\begin{align*}
\phi &= \sum_{h \in \{0,1\}^{2nN}} v_h	(\: \frac{1}{2} Pr[E^T_{\mathbb{P},h}] + \frac{1}{2} Pr[E^H_{\mathbb{P},h} \cap E_{\mathbb{P},h}] \: )\\
&= \sum_{h \in \{0,1\}^{2nN}} v_h	( \: \frac{1}{2} Pr[E^T_{\mathbb{P},h}] + \frac{1}{2} Pr[E^H_{\mathbb{P},h}] Pr[E_{\mathbb{P},h} | E^H_{\mathbb{P},h}] \: ) \\
&= \sum_{h \in \{0,1\}^{2nN}} v_h	( \: \frac{1}{2} (1 - p_{h, T}) + \frac{1}{2} (1 - p_{h, H})(1 - \tilde{p}_h(D^C_{\mathbb{P}, h})) \: )\;.
\end{align*}

Since Lemma 3.1 of \cite{measurement} and Claim \ref{claim:measurement-protocol} taken together yield the inequality
\begin{gather*}
\tilde{p}_h(D_{\mathbb{P}',h}) - \tilde{p}_h(D^C_{\mathbb{P}, h}) \leq \|D^C_{\mathbb{P}, h} - D_{\mathbb{P},h}\|_{TV} \leq \sqrt{p_{h,T}} + p_{h,H} + \mu\;,
\end{gather*}
 it follows that
\begin{align*}
\phi &\leq \sum_{h \in \{0,1\}^{2nN}} v_h	( \: \frac{1}{2} (1 - p_{h, T}) + \frac{1}{2} (1 - p_{h, H})(1 - \tilde{p}_h(D_{\mathbb{P}', h}) + \sqrt{p_{h,T}} + p_{h,H} + \mu) \: ) \\
&\leq \frac{1}{2}\mu + \frac{1}{2} \sum_{h \in \{0,1\}^{2nN}} v_h(1 - p_{h,T} + (1 - p_{h,H})(p_{h,H} + \sqrt{p_{h,T}})) + \frac{1}{2} \sum_{h \in \{0,1\}^{2nN}} v_h(1 - \tilde{p}_h(D_{\mathbb{P}', h})) \\
&\leq \frac{1}{2}\mu + \frac{3}{4} + \frac{1}{2}p\;.
\end{align*}

The upper bound of $\frac{3}{4}$ in the last line can be obtained by straightforward calculation.\footnote{For example, one can obtain this bound by maximising the quantity $f(p_{h,T}, p_{h,H}) = \frac{1}{2}\big(1 - p_{h,T} + (1 - p_{h,H})(p_{h,H} + \sqrt{p_{h,T}})\big)$ under the assumption that $p_{h,T}$ and $p_{h,H}$ lie in $[0,1]$. The function $f$ has one stationary point $(p_{h,T}=\frac{1}{9}, p_{h,H}=\frac{1}{3})$ in $[0,1]^2$; checking $f$ at this point, in addition to its maxima on each of the boundaries of $[0,1]^2$, reveals that the choice of $(p_{h,T}, p_{h,H}) \in [0,1]^2$ which yields the maximum value of $f$ is $(\frac{1}{9},\frac{1}{3})$, giving $f = \frac{2}{3}$. Of course, $\frac{2}{3} < \frac{3}{4}$; we use the bound of $\frac{3}{4}$ for consistency with \cite{measurement}.} We conclude that Protocol \ref{protocol:main-protocol} has a soundness parameter $s$ which is negligibly close to $\frac{3}{4}$.
\end{proof}

\section{Zero-knowledge property of protocol}
\label{section:zk}

In this section, we establish that Protocol~\ref{protocol:main-protocol} is zero-knowledge against arbitrary classical probabilistic polynomial time (PPT) verifiers. Specifically, we show the following:

\begin{lemma}\label{lem:zk}
Suppose that the instance $x=(H,a,b)$ of the $2$-local XZ Hamiltonian problem that is provided as input to the verifier and prover in Protocol \ref{protocol:main-protocol} is a yes-instance, i.e.\ the ground energy of $H$ is smaller than $a$. Then (provided that Assumptions \ref{assumptions} hold) there exists a polynomial-time generated PPT simulator $S$ such that, for any arbitrary PPT verifier $V^*$, the distribution of $V^*$'s final output after its interaction with the honest prover $P$ in Protocol \ref{protocol:main-protocol} is (classical) computationally indistinguishable from $S$'s output distribution.
\end{lemma}

\begin{remark}\label{rk:52}
Lemma~\ref{lem:zk} formulates the zero-knowledge property in terms of classical verifiers and computational indistinguishability against classical distinguishers, because this is the most natural setting for a protocol in which verifier and interaction are classical. However, the same proof can be adapted to show that, for any quantum polynomial-time verifier executing Protocol \ref{protocol:main-protocol}, there exists a quantum polynomial-time generated simulator whose output is QPT indistinguishable in yes-instances from that of the verifier. (In particular, the latter follows from the fact that the second item in Assumptions~\ref{assumptions} implies an NP proof system which is zero-knowledge against quantum polynomial-time verifiers, an implication shown to be true in~\cite{watrous2009zero}.)
\end{remark}

We show that Protocol \ref{protocol:main-protocol} is zero-knowledge by replacing the components of the honest prover with components of a simulator one at a time, and demonstrating that, when the input is a yes-instance, the dishonest verifier's output after each replacement is made is at the least computationally indistinguishable from its output before. The argument proceeds in two stages. In the first, we show that the honest prover can be replaced by a \emph{quantum polynomial-time} simulator that does not have access to the witness $\rho$. In the second, we de-quantise the simulator to show that the entire execution can be simulated by a classical simulator who likewise does not have access to $\rho$. (The latter is desirable because the verifier is a classical entity.)

We begin with the protocol execution between the honest prover $P$ and an arbitrary cheating verifier $V^*$, the latter of whom may take some (classical) auxiliary input $Z_0$, store information (represented by $Z_1$ and $Z_2$), and produce an arbitrary final output $Z_3$. A diagram representing the interaction between $V^*$ and $P$ can be found in Figure \ref{figure:cheating-verifier}.

\subsection{Eliminating the coin-flipping protocol}

Our first step in constructing a simulator is to eliminate the coin-flipping protocol, which is designed to produce a trusted random string $r$, and replace it with the generation of a truly random string. (This step is entirely analogous to step 1 of Section 5 in \cite{qma}, and we omit the analysis.) The new diagram is shown below. In this diagram, \textit{\textsf{coins}} represents a trusted procedure that samples a uniformly random string $r$ of the appropriate length. 

\begin{center}
    \includegraphics[width=0.8\textwidth]{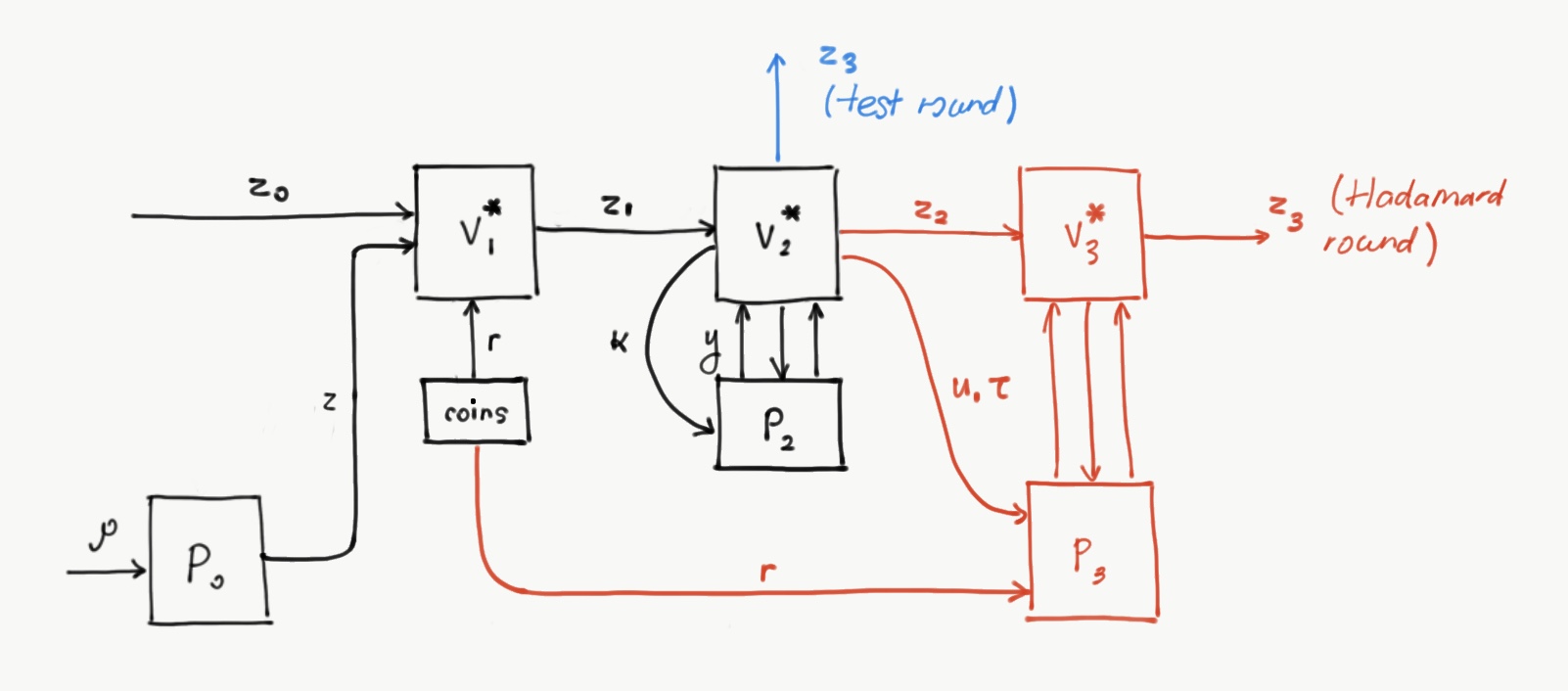}
\end{center}

\subsection{Introducing an intermediary}

Our next step is to introduce an \textit{intermediary}, denoted by $I$, which pretends—to the cheating verifier of Protocol \ref{protocol:main-protocol}—to be its prover $P$, while simultaneously playing the role of \textit{verifier} to the prover from the zero-knowledge proof system of Protocol \ref{protocol:modifiedqma} \footnote{Protocol \ref{protocol:modifiedqma} is identical in structure to the protocol presented in \cite{qma}. We refer the reader to Figure 4 in that paper for a diagram representing the appropriate interactions.}. (We denote the honest prover and honest verifier for the proof system of Protocol \ref{protocol:modifiedqma} by $\mathcal{P}$ and $\mathcal{V}$, respectively, to distinguish them from the prover(s) $P$ and verifier(s) $V$ of the classical-verifier protocol currently under consideration.) We remark, for clarity, that $I$ is a quantum polynomial-time procedure. The essential idea of this section is that $I$ will behave so it is impossible for the classical verifier $V$ to tell whether it is interacting with the intermediary or with its honest prover. (We achieve this simply by making $I$ output exactly the same things that $P$ would.) Given that this is so, the map that $V$ implements from its input to its output, including its auxiliary registers, cannot possibly be different in the previous section as compared to this section.
\begin{figure}[H]
\begin{center}
    \includegraphics[width=\textwidth]{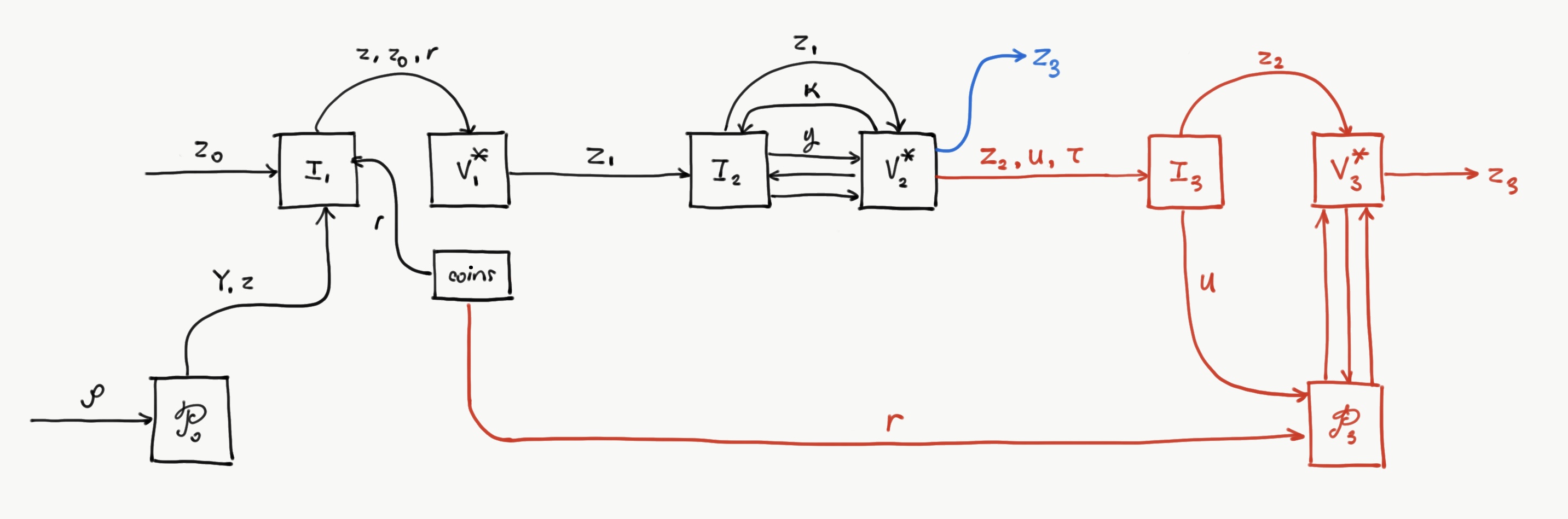}
    \caption{The intermediary interacting with the honest prover from the proof system of Protocol \ref{protocol:modifiedqma}, denoted by $\mathcal{P}$, and also with the cheating classical verifier $V^*$. $I_1$ receives the encoded quantum witness, which we have denoted by $Y$, from $\mathcal{P}$, in addition to $\mathcal{P}$'s commitment $z$. It then sends $z$ to $V^*_1$, along with $Z_1$, the auxiliary input that $V^*_1$ is supposed to receive, and $r$, the random string generated by \textit{\textsf{coins}}. $I_2$ passes on any output $V^*_1$ produces to $V^*_2$, performs itself the procedure for committing to a quantum state from \cite{measurement}, and executes the measurement protocol with $V^*_2$. $I_3$ receives the measurement outcomes $u$ and the trapdoors $\tau$ from $V_2^*$, and checks whether the trapdoors are valid. If they are invalid, it aborts directly; if they are valid, it sends $u$ on to $\mathcal{P}_3$ and passes $Z_2$ to $V^*_3$, so that $\mathcal{P}_3$ and $V^*_3$ can execute the NP zero-knowledge proof protocol. (Each part of $I$ should also send everything it knows to its successor, but we have omitted these communications for the sake of cleanliness, as we omitted the communication between parts of the prover in previous diagrams.)}
    \label{figure:intermediary}
\end{center}
\end{figure}

\subsection{Simulating the protocol with a quantum simulator}

We now note that Figure \ref{figure:intermediary} looks exactly like Figure 4 from \cite{qma}, if we consider the intermediary $I$ and the cheating classical verifier $V^*$ taken together to be a cheating verifier $\mathcal{V}'$ for the proof system of Protocol \ref{protocol:modifiedqma}.

\begin{figure}[H]
\begin{center}
    \includegraphics[width=\textwidth]{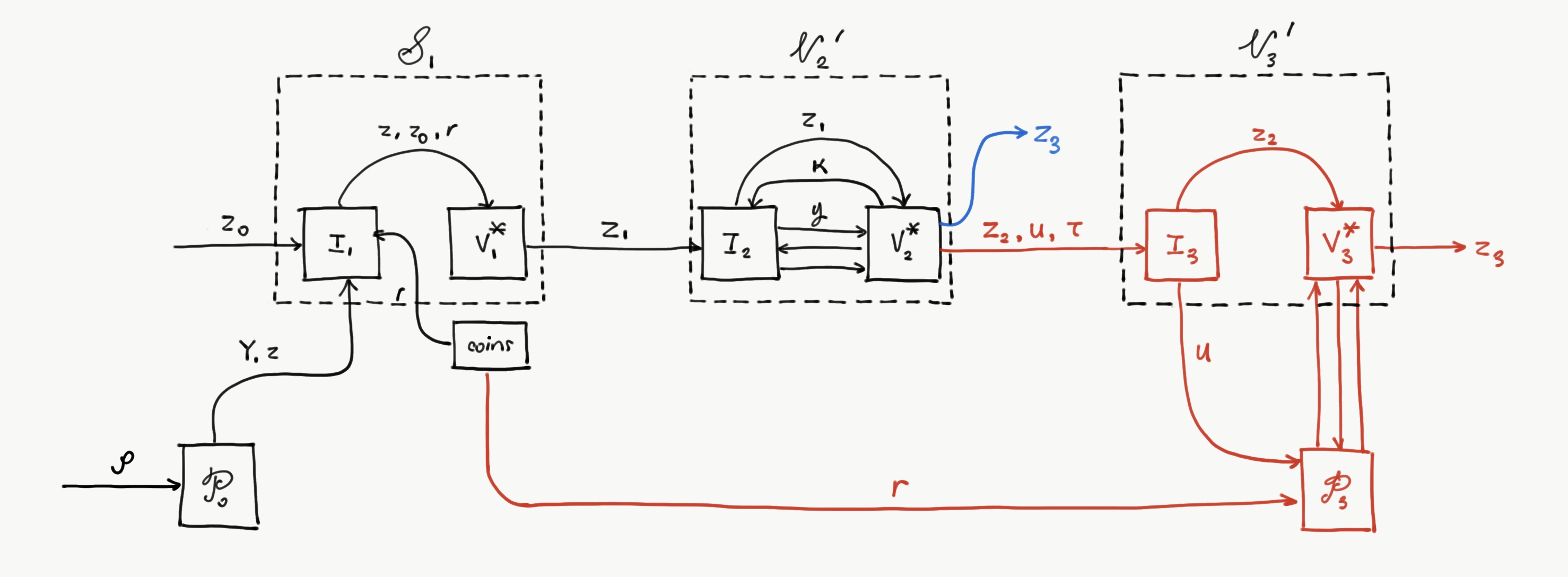}
    \caption{Compare to Figure 4 of \cite{qma}. Note that $\mathcal{S}_1$ includes the behaviour of an arbitrary $\mathcal{V}'_1$; the reason it is called $\mathcal{S}_1$ and not $\mathcal{V}'_1$ is because $\mathcal{V}'_1$ obtains $r$ from a coin-flipping protocol, while $\mathcal{S}_1$ generates $r$ using \textit{\textsf{coins}}. In all other respects, $\mathcal{S}_1$ is the same as $\mathcal{V}'_1$.}
    \label{figure:simulation-ready-1}
\end{center}
\end{figure}

Using similar reasoning as in \cite{qma} (and recalling that, by Lemma~\ref{lem:modified-qma-proof}, it still works when the Hamiltonian being verified is an XZ Hamiltonian), therefore, we conclude that we can replace $\rho$ in Figure \ref{figure:simulation-ready-1} with $\rho_r$—where $\rho_r$ is a quantum state specifically designed to pass the challenge indexed by $r$—without affecting the verifier's output distribution (to within computational indistinguishability). See Remark \ref{remark:rho_r} for a procedure that explicitly constructs $\rho_r$. Note that, if our objective was to achieve a \textit{quantum} simulation without knowing the witness state $\rho$, our task would already be finished at this step. However, our verifier is classical; therefore, in order to prove that our classical verifier's interaction with its prover does not impart to it any knowledge (apart from the fact that the problem instance is a yes-instance) that it could not have generated itself, we need to achieve a \emph{classical} simulation of the argument system.

\subsection{Simulating the protocol with a classical simulator}

\subsubsection{Replacing $\mathcal{P}_0$ and $I_1$}

If we want to simulate the situation in Figure \ref{figure:simulation-ready-1} classically, then we need to de-quantise $\mathcal{P}_0$, $I_1$ and $I_2$. ($I_3$ and $\mathcal{P}_3$ are already classical.) Our first step is to replace $\mathcal{P}_0$ and $I_1$ with a single classical entity, $I_1'$.

$I_1'$ simply chooses encoding keys $(t, \pi, a, b)$ and generates $z$, a commitment to the encoding keys $(\pi, a, b)$. It then sends $z, r$ and $Z_1$ to $V^*_1$, as $I_1$ would have. Because $I_1'$ has exactly the same output as $I_1$, the verifier's output in Figure \ref{figure:I_1-prime} is the same as its output in Figure \ref{figure:simulation-ready-1}. (We assume that the still-quantum $I_2$ now generates $\rho_r$ for itself.)

\begin{figure}[H]
\includegraphics[width=\textwidth]{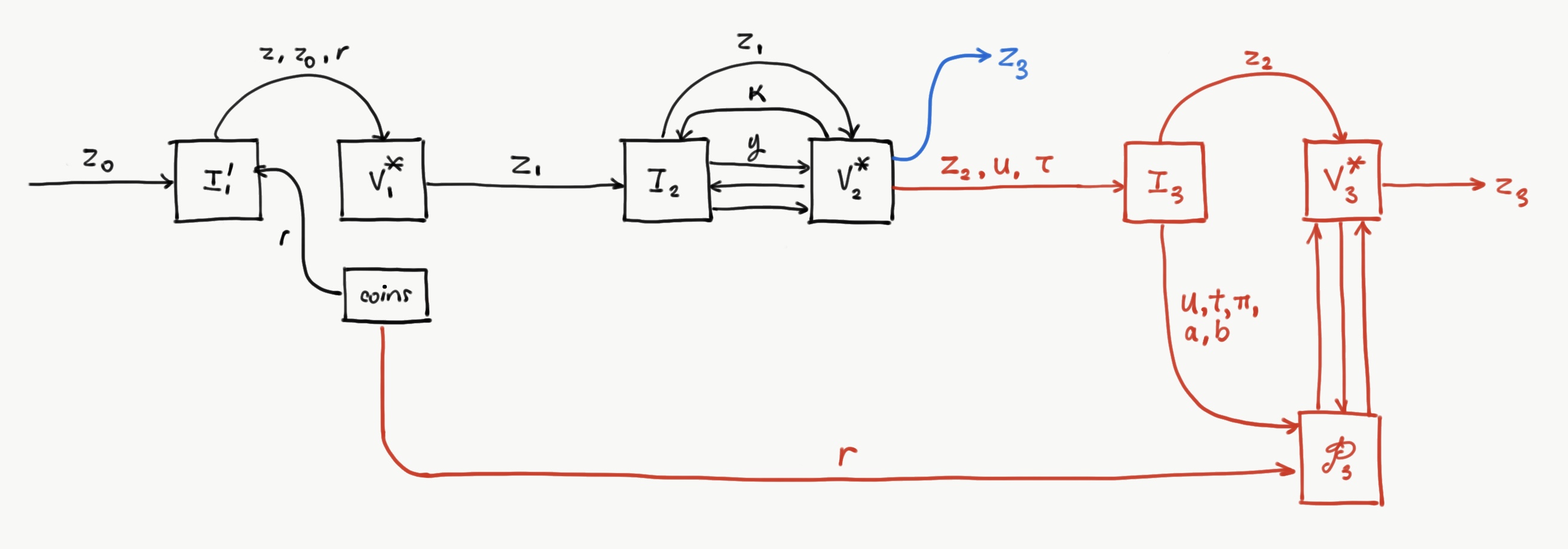}
\caption{$\mathcal{P}_0$ and $I_1$ have been replaced by $I_1'$.}
\label{figure:I_1-prime}
\end{figure}

\subsubsection{Some simplifications (which make it possible to de-quantise $I_2$)}

Following \cite{qma}, we make some alterations to Figure \ref{figure:I_1-prime} that will allow us to eventually de-quantise $I_2$. The alterations are as follows:

\begin{enumerate}[itemsep=1pt, topsep=0pt]
\item Replace $V_3^*$ and $\mathcal{P}_3$ with an efficient simulation $S_3$. (An efficient simulation of the NP proof protocol execution between $V_3^*$ and $\mathcal{P}_3$ is guaranteed to exist because the NP proof protocol is zero-knowledge.) Recall that the statement $\mathcal{P}_3$ is meant to prove to $V^*_3$ in a zero-knowledge way is as follows: `There exists a string \(s_p\) and an encoding key \((t, \pi, a, b)\) such that \(z = \mathsf{commit}((\pi, a, b), s_p) \) and \(Q(t, \pi, a, b, r, u) = 1\).' The zero-knowledge property of the NP proof system guarantees that, for yes-instances, the output of $S_3$ is indistinguishable from the output of the protocol execution between $V_3^*$ and $P_3$. In our case, $I_1'$ always holds $s_p$ and $(\pi, a, b)$ such that \(z = \mathsf{commit}((\pi, a, b), s_p)\), and the honest prover will abort the protocol if \(Q(t, \pi, a, b, r, u) = 0\). Therefore, whenever the prover does not abort, the output of $S_3$ is computationally indistinguishable from that of $V_3^*$ and $P_3$. We assume, following \cite{qma}, that $S_3$ also behaves as $V_3^*$ would when the prover aborts. If it does, then Figure \ref{figure:NP-proof-simulation} is computationally indistinguishable from Figure \ref{figure:I_1-prime}.

\begin{figure}[H]
\includegraphics[width=\textwidth]{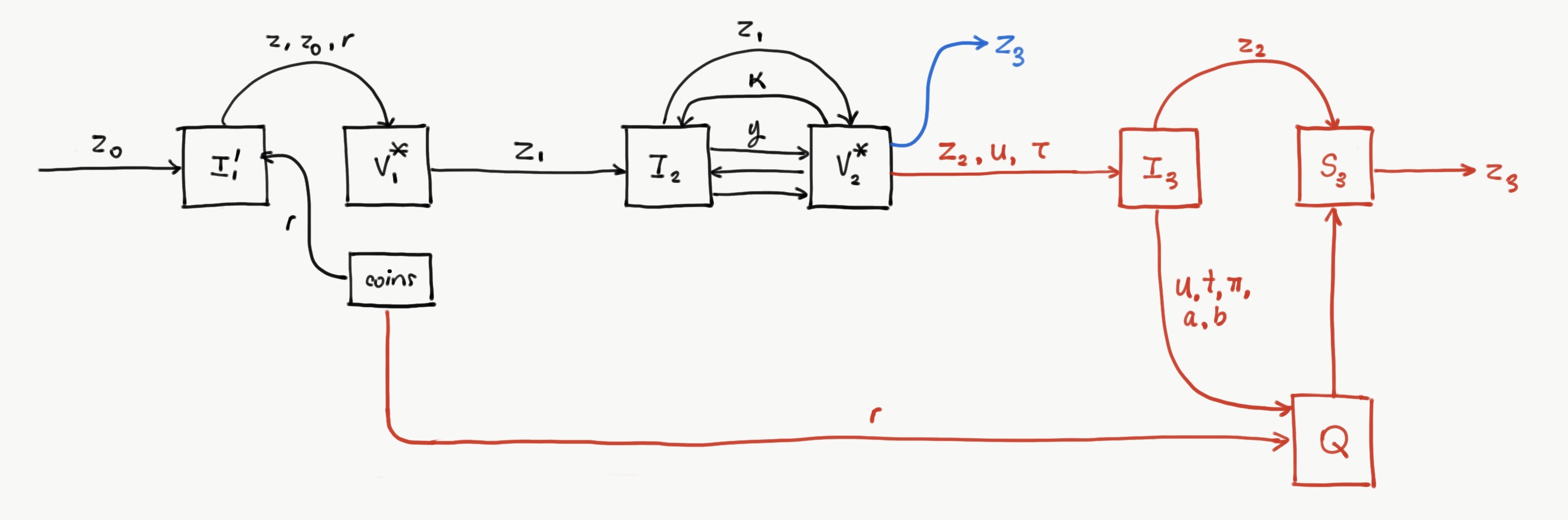}
\caption{$V_3^*$ and $\mathcal{P}_3$ have been replaced by $S_3$. Note that $S_3$ does not require access to the witness $(s_p, t, \pi, a, b)$, and so $s_p$ can be discarded immediately after $I_1'$ is run.}
\label{figure:NP-proof-simulation}
\end{figure}

\item Replace the generation of the genuine commitment $z$ with the generation of a commitment $z' = \textsf{commit}((\pi_0, a_0, b_0), s_p)$, where $\pi_0, a_0$ and $b_0$ are fixed strings independent of the encoding key $(t, \pi, a, b)$ that $I_1'$ chooses. Because the commitment protocol is (computationally) concealing, and the commitment is never opened (recall that $s_p$ is discarded after $I_1'$ is run), $V_1^*$ should not be able to tell (computationally speaking) that $z$ has been replaced by $z'$.

The genuine encoding key is still used to evaluate the predicate $Q$. Note that, because $z$ has been replaced with $z'$, the statement for which $S_3$ must simulate the execution of a zero-knowledge proof between $V_3^*$ and $\mathcal{P}_3$ is now as follows: `There exists a string \(s_p\) and an encoding key \((t, \pi, a, b)\) such that \(z' = \mathsf{commit}((\pi, a, b), s_p) \) and \(Q(t, \pi, a, b, r, u) = 1\).' This statement is, in general, no longer true, because the commitment protocol is perfectly binding. However, if the predicate $Q$ is still satisfied for the encoding key $(t, \pi, a, b)$ that $I_3$ sent, then $S_3$ will proceed to generate a transcript for the no-instance that is computationally indistinguishable from a transcript for a yes-instance. If $Q$ is no longer satisfied, then $S_3$ will abort, as before. In effect, therefore, the cheating verifier $V^*$ will not be able to tell (up to computational indistinguishability) that $z$ has been replaced by $z'$, and that the NP statement being `proven' to it is no longer true.
\end{enumerate}

\subsubsection{De-quantising $I_2$}
\label{sec:iprimesim}

We now replace $I_2$ with a classical entity $I_2'$. In the process, we require modifications to the behaviour of $I_3$.

Knowing $r$, $I_2'$ can calculate for itself what $\rho_r$ should be, though it cannot physically produce this state. As we noted in Remark \ref{remark:rho_r}, $\rho_r$ is a simple state: it is merely the tensor product of $\ket{0}, \ket{1}, \ket{+}$ and $\ket{-}$ qubits. Applying the concatenated Steane code to $\rho_r$ will then result in a tensor product of $N$-qubit states that look like
\begin{gather}
\sum_{x \in \mathcal{D}^0_N} \ket{x},  \quad
\sum_{x \in \mathcal{D}^1_N} \ket{x},  \quad
\sum_{x \in \mathcal{D}^0_N} \ket{x} + \sum_{x \in \mathcal{D}^1_N} \ket{x}, \quad \text{and} \quad
\sum_{x \in \mathcal{D}^0_N} \ket{x} - \sum_{x \in \mathcal{D}^1_N} \ket{x}, \label{eq:0}
\end{gather}
after appropriate normalisation.

A brief argument will suffice to establish that it is possible to classically simulate standard or Hadamard basis measurements on the qubits in $E(\rho_r)$. Each qubit of $E(\rho_r)$ is either an encoding qubit or a trap qubit, up to the application of a random single-qubit Pauli operator. Simulating standard-basis measurements of encoding qubits is classically
 feasible, because $\mathcal{D}^0_N$ and $\mathcal{D}^1_N$ are polynomially sized, and the expressions in~\eqref{eq:0} only involve superpositions over those sets with equal-magnitude coefficients. Simulating standard-basis measurements of trap qubits, which are always initialised either to $\ket{0}$ or $\ket{+}$, is trivially feasible.

To simulate a Hadamard basis measurement, we can take advantage of the transversal properties of the encoding scheme, and apply $H$ before we apply the concatenated Steane code. Denote the application of the concatenated Steane code to $\rho_r$ by $S(\rho_r)$. We have that
\[S(H^{\otimes n} \rho_r H^{\otimes n}) = H^{\otimes nN} S(\rho_r) H^{\otimes nN}\]
by transversality. To simulate a Hadamard basis measurement of $E(\rho_r)$, we then
\begin{enumerate}[itemsep=1pt, topsep=0pt]
\item Apply $H^{\otimes n}$ to $\rho_r$. This is easy to classically simulate, because $\rho_r$ is a tensor product of $\ket{0}, \ket{1}, \ket{+}$ and $\ket{-}$ qubits.
\item Apply the concatenated Steane code to $H^{\otimes n} \rho_r H^{\otimes n}$. Simulating this is classically feasible, by the same argument that we used for standard basis measurements, because $H^{\otimes n} \rho_r H^{\otimes n}$ is still a tensor product of $\ket{0}, \ket{1}, \ket{+}$ and $\ket{-}$ qubits.
\item Concatenate trap qubits to each $N$-qubit block in $S(H^{\otimes n} \rho_r H^{\otimes n}) = H^{\otimes nN} S(\rho_r) H^{\otimes nN}$. Simulate the application of $H$ to each trap qubit (which is, once again, classically easy to do because each trap qubit is initialised either to $\ket{0}$ or to $\ket{+}$).
\item Apply the permutation $\pi$ to each $2N$-tuple.
\item Simulate a standard basis measurement of the result.
\item XOR the string $b$ to the measurement outcome ($b$ was previously the $Z$-key for the Pauli one-time pad).
\end{enumerate}
Having established that it is possible to classically simulate standard and Hadamard basis measurements of the qubits in $E(\rho_r)$, we now describe the procedure that the classical $I_2'$ should follow for each qubit $i$ in the state $E(\rho_r)$.
\begin{enumerate}[itemsep=1pt, topsep=0pt]
\item During the commitment phase, $I_2'$ simulates a standard basis measurement on the $i$th qubit, obtains a simulated measurement result $\beta_i$, and then chooses a uniformly random preimage $x_i$ from the domain of the function specified by $\kappa_i$. It applies the function specified by $\kappa_i$ to $\beta_i \| x_i$ and sets $y_i = \eta_{\kappa_i}(\beta_i \| x_i)$.
\item If the verifier requests a test round, $I_2'$ sends $\beta_i \| x_i$ to the verifier. This is exactly what the quantum prover $I_2$ would send in the case of a test round, so the verifier cannot tell that it is  interacting with $I_2'$ instead of $I_2$.
\item If the verifier requests a Hadamard round, $I_2'$ sends a uniformly random string $s_i \in \{0,1\}^{w+1}$ to the verifier, where $w$ is the length of the preimages. 
In the same situation, the quantum $I_2$ would have sent Hadamard basis measurements of the $w+1$ qubits in the $i$th committed qubit in $E(\rho)$ and its associated preimage register.

\begin{enumerate}[itemsep=1pt, topsep=0pt]
\item If $h_i = 0$, the outcomes of these measurements are uniformly distributed and thus indistinguishable from the distribution of strings $s_i$ reported by $I'_2$.
\item Let $\ket{\psi_i}$ be the state of the $i$th qubit of $E(\rho)$, let $x_{0,i}$ and $x_{1,i}$ be the two preimages to $y_i$ under the function $f_{\kappa_i}$, and let $b_i$ be the $i$th bit of the one-time-pad $Z$-key $b$ from $I_1'$'s encoding key $(t, \pi, a, b)$. If $h_i = 1$, the outcomes of $I_2$'s Hadamard basis measurements can be represented as a tuple $(\beta_i,d_i)$, where $d_i$ is uniformly random, and
\begin{gather*}
\beta_i = d_i \cdot (x_{0,i} \oplus x_{1,i}) \oplus b_i \oplus \textsf{Meas}(H\ket{\psi_i})\;.
\end{gather*}

($\textsf{Meas}$ here denotes a standard basis measurement.)

Note that the distribution over $(b_i, \beta_i, d_i)$ which one would obtain by measuring $\ket{\psi_i}$ in the Hadamard basis, choosing $d_i$ and $b_i$ uniformly at random, and letting
\begin{gather*}
\beta_i = d_i \cdot (x_{0,i} \oplus x_{1,i}) \oplus b_i \oplus \textsf{Meas}(H\ket{\psi_i})
\end{gather*}
is equivalent to the one that one would obtain choosing a uniformly random $s_i$, measuring $\ket{\psi_i}$ in the Hadamard basis, calculating
\begin{gather*}
b_i = s_{i,1} \oplus d_i \cdot (x_{0,i} \oplus x_{1,i}) \oplus \textsf{Meas}(H\ket{\psi_i})\;,
\end{gather*}
and finally setting $\beta_i = s_{i,1}, \: d_i = s_{i,2} \cdots s_{i,w+1}$.

The former set of actions is equivalent to the set of actions that $I_2$ performs. The latter set of actions is (as we will shortly show) classically feasible provided that we have the verifier's trapdoors. Note that $I_2'$ only needs to send the verifier $s_i$, and can rely on its successor $I_3$, who \emph{will} have access to the verifier's trapdoors, to calculate the bits $b_i$ retroactively. It follows that, given that $I_3$ can produce correct bits $b_i$ (we will shortly show that it can), the distribution of strings reported by $I'$ is identical to the distribution of outcomes reported by $I$.
\end{enumerate}
\end{enumerate}

Having established that $I_2'$ and $I_2$ are the same from $V_2^*$'s perspective (meaning that it must have the same behaviour that it did in Figure \ref{figure:I_1-prime} after $I_2$ is replaced with $I_2'$), it remains to ensure that the choice of the one-time pad $Z$-key $ b$ is consistent with the $s_i$ that $I_2'$ picked. We relegate the task of making this choice  to $I_3'$, our new version of $I_3$, because it has access to the verifier's trapdoors $\tau$. If any of the trapdoors that it receives from the verifier are invalid, or if any of the ETCFF keys $\kappa$ which the verifier chose are invalid, $I_3'$ aborts, as specified in Protocol \ref{protocol:main-protocol}. (`Validity', here, means the following: 1) all the $\kappa$s which the verifier sent earlier well-formed, and 2) for each $y_i$, the trapdoor $\tau_{\kappa_i}$ correctly inverts the function specified by $\kappa_i$. We expand on this notion of `validity' in Appendix \ref{appendix:trapdoor-check}.) Presuming upon valid keys and valid trapdoors, $I_3'$ then deduces the verifier's choices of measurement basis, $h$, from $\tau$. Given that the trapdoors are valid and that the keys are well-formed, $I_3'$ can be confident that its deductions in this regard will lead it to behave in the same way that the honest prover would, because (given valid keys and trapdoors) $I_3'$ will know exactly which superpositions the honest prover would have obtained during the measurement protocol after following the verifier's instructions.

For notational convenience, let $\ket{\psi^*}$ denote the state obtained by applying the first three steps of $E$, but not the last step, to $\rho_r$. $I_3'$ subsequently executes the following procedure for all $i$ such that $h_i = 1$:

\begin{enumerate}[itemsep=1pt, topsep=0pt]
\item Set $d_i$ to be the last $w$ bits of $s_i$, and compute $d_i \cdot (x_{0,i} \oplus x_{1,i})$ using the trapdoor $\tau_{\kappa_i}$.
\item Simulate a standard basis measurement of $H X^{a_i} \ket{\psi^*_i}$. Denote the result by $\beta_i$. (Here, $a_i$ refers to the $i$th bit of $a$, where $a$ is taken from $I_1'$'s initial choice of one-time pad keys, and $\ket{\psi^*_i}$ denotes the $i$th qubit of $\ket{\psi^*}$.)
\item Set $b'_i$ (the $i$th bit of $b'$, the new $Z$-key for the one-time pad) to be equal to $\beta_i \oplus s_{i,1} \oplus d_i \cdot (x_{0,i} \oplus x_{1,i})$ (where $s_{i,1}$ refers to the first bit of $s_i$). This will cause the equation $\textsf{Meas}(H\ket{\psi_i}) \oplus d_i \cdot (x_{0,i} \oplus x_{1,i}) = s_{i,1}$ to be satisfied:
\begin{gather*}
\textsf{Meas}(H\ket{\psi_i}) \oplus d_i \cdot (x_{0,i} \oplus x_{1,i}) = s_{i,1} \\
\iff \textsf{Meas}(H Z^{b_i} X^{a_i} \ket{\psi^*_i}) \oplus d_i \cdot (x_{0,i} \oplus x_{1,i}) = s_{i,1} \\
\iff b_i \oplus \textsf{Meas}(H X^{a_i} \ket{\psi^*_i}) \oplus d_i \cdot (x_{0,i} \oplus x_{1,i}) = s_{i,1} \\
\iff b_i \oplus \beta_i \oplus d_i \cdot (x_{0,i} \oplus x_{1,i}) = s_{i,1} \\
\iff b_i = \beta_i \oplus s_{i,1} \oplus d_i \cdot (x_{0,i} \oplus x_{1,i}).
\end{gather*}

\end{enumerate}

Having done this, $I_3'$ then feeds $(t, \pi, a, b')$ into $Q$. (Note that replacing $b$ with $b'$ cannot create any conflict with the commitment string $z'$ that the verifier will notice, because $z'$ was already independent of the one-time-pad keys $(a,b)$.) In all other respects $I_3'$ behaves the same way that $I_3$ did.

The final simulation will be as follows:

\begin{figure}[H]
\includegraphics[width=\textwidth]{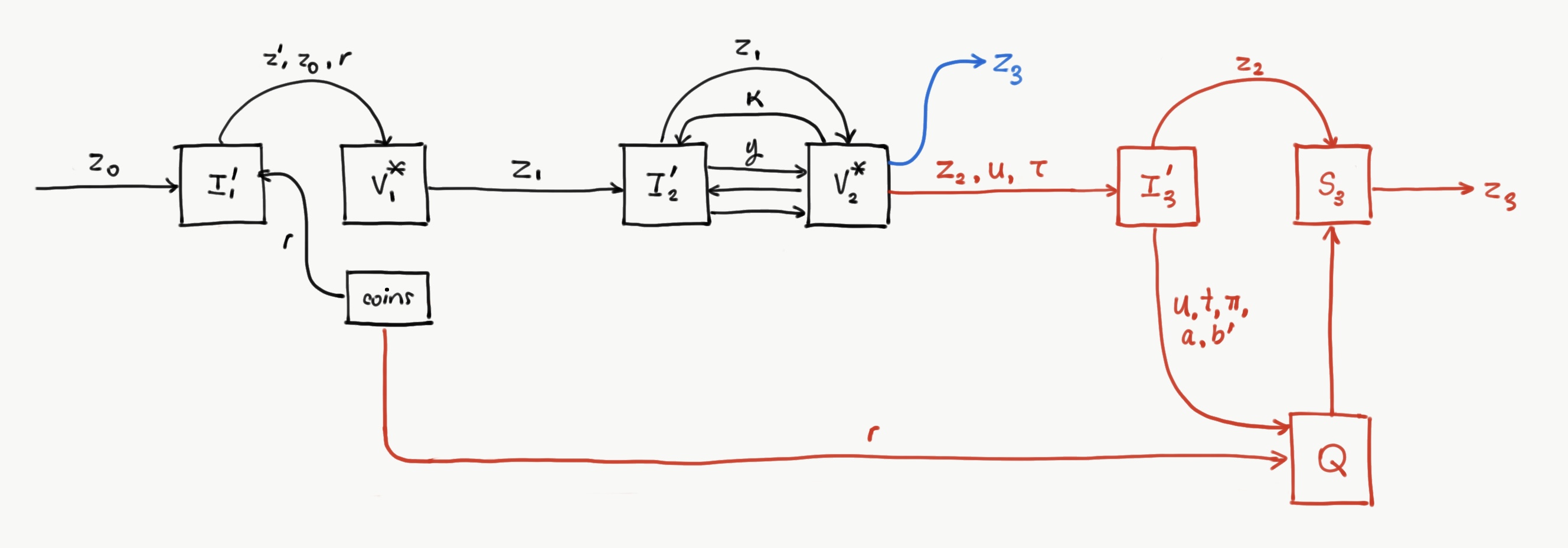}
\end{figure}

Since all the entities in this simulation are classical and efficient, and none have access to information about the witness state $\rho$, it follows that the protocol is zero-knowledge.

\appendix

\section{LWE-based ETCFF family and efficient trapdoor check}
\label{appendix:trapdoor-check}

In order to explain the trapdoor check that the honest prover of Protocol \ref{protocol:main-protocol} implements during step 8 of Protocol~\ref{protocol:main-protocol}, we briefly outline, at a level of detail appropriate for us, how the LWE-based ETCFF family that is used in~\cite{measurement} is constructed. 

We begin by introducing the instantiations of the keys $\kappa$ and the trapdoors $\tau$ for noisy trapdoor claw-free ($f$) and trapdoor injective ($g$) functions, whose properties we have relied upon in a black-box way for the rest of this work. For details, we refer the reader to Section 9 of \cite{measurement}.

The key $(\kappa_1, \kappa_2)$ for an noisy two-to-one function $f$ is $(A, As + e)$, where $A$ is a matrix in $\mathbb{Z}_q^{n\times m}$ and $e\in \mathbb{Z}_q^n$ is an error vector such that $|e| < B_f$ for some small upper bound $B_f$. (The specific properties that $B_f$ should satisfy will be described later.) Here, $n,m$ are integers, and $q$ is a prime power modulus that should be chosen as explained in~\cite{measurement}. In addition, in order to implement the trapdoor, we assume that the matrix $A$ is generated using the efficient algorithm $\textsf{GenTrap}$ which is described in Algorithm 1 in~\cite{trapdoors}. (For convenience, we use the `statistical instantiation' of the procedure described in Section 5.2 of the paper.) $\mathsf{GenTrap}$ produces a matrix $A$ that has the form $A=[\overline{A}|HG-\overline{A}R]$, for some publicly known $G \in \mathbb{Z}_q^{m \times w}$, $n \leq w \leq m$, some $\overline{A} \in \mathbb{Z}_q^{n \times (m-w)}$, some invertible matrix $H \in \mathbb{Z}_q^{n \times n}$, and some $R \in \mathbb{Z}^{(m-w) \times w}$, where $R=\tau_A$ is the trapdoor to $A$. As shown in~\cite[Theorem 5.4]{trapdoors}, it is straightforward, given the matrix $R$, to verify that $R$ is a `valid' trapdoor, in the sense that it allows a secret vector $s$ to be recovered from a tuple of the form $(A, b=As+e)$ with certainty when $e$ has magnitude smaller than some bound $B_{\textsf{Invert}}$. Checking that $R$ is a valid trapdoor involves computing the largest singular value of $R$ and checking that $A$ is indeed of the form $A=[\overline{A}|HG-\overline{A}R]$ for some invertible $H$ and for the publicly known $G$. Using any valid trapdoor, recovery can be performed via an algorithm $\textsf{Invert}$ described in~\cite{trapdoors}. 

 The key for an injective function $g$, meanwhile, is $(A, u)$, where $u$ is a random vector not of the form $As + e$ for any $e$ of small enough magnitude. (Again, `small enough' here refers to a specific upper bound, and what the bound is precisely will be described later. The distribution of $u$ is uniform over all vectors that satisfy this latter requirement.) The trapdoor $\tau_A$ is still the $R$ corresponding to the matrix $A$ which is described in the preceding paragraph. 

The functions $f_{\kappa}$ and $g_{\kappa}$ both take as input a bit $b$ and a vector $x$ and output \emph{a probability distribution} (to be more precise, a truncated Gaussian distribution of the kind defined in Section 2.3, equation 4 of \cite{measurement}). We clarify that, when we say that the functions output a probability distribution, we mean that they should be thought of as maps from the space of strings to the set of probability distributions, not that their outputs are randomised. Given a sample $y$ from one such probability distribution $Y$, the trapdoor $\tau_A$ can be used to recover the tuple(s) $(b, x)$ which are preimages of $y$ under the function specified by $\kappa$. (See Definition \ref{def:preimage-of-y} for a definition of the phrase `preimage of $y$'.) The functions $f_\kappa$ and $g_\kappa$ can be defined (using notation explained in the paragraph below the definition) as follows:

\begin{definition}[Definition of trapdoor claw-free and trapdoor injective functions]
\label{def:f-and-g}
\begin{align*}
\textbf{(a)} \quad f_{\kappa}(b, x) &= Ax + e_0 + b \cdot (As + e)\;, \\
&\quad \text{where $e_0$ is distributed as a truncated Gaussian with bounded magnitude $|e_0|_{max}$} \\
\textbf{(b)} \quad g_{\kappa}(b, x) &= Ax + e_0 + b \cdot u\;.
\end{align*}
\end{definition}

What the above notation means is that one \emph{samples from} the distribution determined by the input $(b,x)$ and the function key $\kappa = (\kappa_1, \kappa_2)$ by sampling $e_0$ from a truncated Gaussian centred at the origin and then computing $\kappa_1 x + e_0 + b \cdot (\kappa_2)$. A key feature of the $f$ functions is that the output distributions given by $f_\kappa(0, x)$ and $f_\kappa(1, x - s)$ are truncated Gaussians which overlap to a high degree (so that the statistical distance between the distributions $f_\kappa(0, x)$ and $f_\kappa(1, x - s)$ is negligible). The $g$ functions, meanwhile, are truly injective in the sense that $g(b, x)$ and $g(b', x')$ never overlap for $(b, x) \neq (b', x')$. In order that these two things are true, we require that the $e$ in Definition \ref{def:f-and-g}(a) is very small ($B_f \ll |e_0|_{max}$), and that the $u$ in Definition \ref{def:f-and-g}(b) is such that $u \neq As + e$ for any $|e| < B_g$, where $B_g > |e_0|_{max}$. It follows from hardness of the (decisional) LWE assumption that the keys for the $f$ functions and the keys for the $g$ functions are computationally indistinguishable.

The trapdoor check that the prover of protocol \ref{protocol:main-protocol} executes in step 8 is as follows:
\begin{protocol}[Trapdoor and key check]
\label{protocol:trapdoor-and-key-check}
\item[]Let $\kappa_i = (A_i, A_i s_i + e_i)$. (Note that $|e_i|$ need not be smaller than any particular bound in this definition of $\kappa_i$.) For all $i \in \{1, \ldots, 2nN\}$:
\begin{enumerate}
\item Check that $\tau_{\kappa_i}$ is a `valid' trapdoor for $A_i$, in the sense that was explained in the third paragraph of this appendix. If it is not, abort.
\item For a choice of $B_f$, $B_g$ and $|e_0|_{max}$ such that $B_f \ll |e_0|_{max} < B_g \leq B_{\textsf{Invert}}$ and $B_g - |e_0|_{max} > |e_0|_{max}$, check that one of the following three conditions hold:
\begin{enumerate}
\item \textsf{Invert} applied to $\kappa_{2,i}$ succeeds, and recovers an $e$ such that $|e| < B_f$, \emph{or}
\item \textsf{Invert} applied to $\kappa_{2,i}$ succeeds, and recovers an $e$ such that $B_g < |e|$, \emph{or}
\item \textsf{Invert} fails.
\end{enumerate}
\end{enumerate}
\end{protocol}

\begin{figure}[H]
\begin{center}
\includegraphics[width=0.5\textwidth]{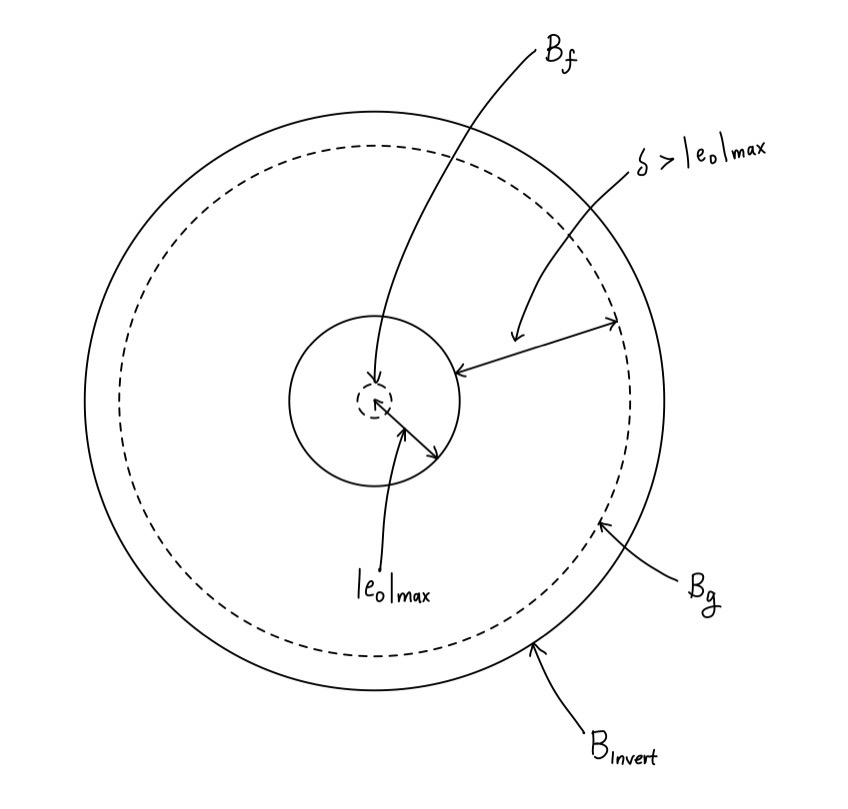}
\end{center}
\caption{Diagram illustrating one possible choice of parameters that satisfies the conditions in 2. above. When a circle is labelled with a number (such as $B_f$ or $B_g$), the radius of the circle represents the size of that number.}
\end{figure}

The conditions in step 2 above are intended to ensure that, for all $i$, $\kappa_i$ is a key either for an $f$ or a $g$ function, and therefore well-formed. An ill-formed key would be of the form $\kappa_{bad} = As + e$ for $B_f < e < B_g$; for some choices of $B_f$ and $B_g$, a subset of $\kappa_{bad}$s as just defined would behave neither like keys for $f$ functions nor like keys for $g$ functions, because the distributions $\eta_{\kappa_{bad}}(0,x)$ and $\eta_{\kappa_{bad}}(1, x-s)$ would overlap but not to a sufficient degree. The specifications on the parameters that are made in step 2 above, and the tests prescribed for the prover, are designed to ensure that $B_f$ and $B_g$ are properly chosen and that the prover can check efficiently that the verifier's $\kappa$s are well-formed according to these appropriate choices of $B_f$ and $B_g$.

The following claim shows that, given a valid trapdoor (i.e.\ a matrix $R$ that satisfies the efficiently verifiable conditions described in the third paragraph of this appendix), and a well-formed key $\kappa$, the trapdoor can be used to successfully recover all the preimages under a function $\eta_\kappa$ to any sample $y$ from a distribution $Y$ in the range of the function $\eta_\kappa$. This claim is needed to justify the correctness of the ``de-quantized'' simulator $I'_2$ considered in Section~\ref{sec:iprimesim}: if $I_2'$ can be sure that it has recovered all the preimages to $y$, and no others, then it can successfully simulate the honest prover.

\begin{claim}
\label{claim:recover-all-preimages}
Let $A$ be a matrix in $\mathbb{Z}_q^{n \times m}$, let $\kappa = (A, \kappa_2)$, and let the function $\eta_\kappa$ be defined by $\eta_\kappa(b, x) = Ax + e_0 + b \cdot \kappa_2$. (The output of $\eta_\kappa$ is, as in Definition \ref{def:f-and-g}, a probability distribution.) Let $\tau_A$ be a purported trapdoor for $A$. Suppose that $\kappa$ passes the test in step 2 of Protocol \ref{protocol:trapdoor-and-key-check}, and suppose that the trapdoor $\tau_A$ inverts the matrix $A$, in the sense that, given $r = As + e$ for some $e$ of sufficiently small magnitude, $\tau_A$ can be used to recover the unique $(s, e)$ such that $As + e = r$. Then one can use $\tau_A$ to efficiently recover all the preimages to any $y$ sampled from any distribution $Y$ in the range of $\eta_\kappa$.
\end{claim}

\begin{proof}
By hypothesis, $\kappa_2$ is either of the form $As + e$ for some $(s, e)$ (with $|e| < B_f$), or it is not of the form $As + e$ for any $e$ such that $|e| < B_g$. We do not know \emph{a priori} which of these is the case, but the procedure that we perform in order to recover the preimage(s) to $y$ is the same in both cases:
\begin{enumerate}
\item Use the trapdoor $\tau_A$ to attempt to find $(x_1, e_1)$ such that $Ax_1 + e_1 = y$. If such an $(x_1, e_1)$ exists, and $|e_1| < |e_0|_{max}$, record $0 \| x_1$ as the first preimage.
\item Use the trapdoor $\tau_A$ to attempt to find $(x_2, e_2)$ such that $Ax_2 + e_2 = y - \kappa_2$. If such an $(x_2, e_2)$ exists, and $|e_2 < |e_0|_{max}$, record $1 \| x_2$ as the second preimage.
\end{enumerate}

If $\kappa_2 = As + e$ for some $s$ and $e$ such that $|e| \ll B_f < |e_0|_{max}$, then this procedure will return two preimages (except with negligible probability, which happens when $y$ comes from the negligibly-sized part of the support of a distribution $f_\kappa(b,x)$ which is not in the support of the distribution $f_\kappa(\lnot b, \: x + (-1)^b s)$; this can occur if $y$ is a sample such that $|e_0| + |e| > |e_0|_{max}$, using notation from Definition \ref{def:f-and-g}). Assuming that the latter is not the case, in step 1, the algorithm above will recover $x$ such that $y = Ax + e_0$ for some $e_0 < |e_0|_{max}$, because (under our assumption, and by linearity) $y$ is always of the form $Ax + e_0$. In step 2, it will recover $x' = x - s$, because $x' = x - s$ will satisfy the equation $y - (As + e) = Ax' + e'$ for $e' = e_0 - e$, and $|e_0 - e| < |e_0| + |e| < |e_0|_{max}$. We know that $y$ has two preimages under our assumption, so we conclude that, when our assumption holds, the algorithm returns all of the preimages to $y$ under $\eta_\kappa$ and no others. In the negligible fraction of cases when $y$ has only one preimage even though $\kappa_2 = As + e$, the algorithm returns one preimage, which is also the correct number.

It can be seen by similar reasoning that, when $\kappa_2 = u$ for $u \neq As + e$ for any $e$ such that $|e| < B_g$, this procedure will return exactly one preimage, which is what we expect when $\kappa_2 = u$.
\end{proof}

In the context of Protocol \ref{protocol:main-protocol}, the honest prover knows that $\eta_{\kappa_i}$ has been evaluated correctly for all $i$, because the prover evaluated these functions for itself. Therefore, given Claim \ref{claim:recover-all-preimages}, if our goal is to show that the honest prover can efficiently determine whether or not a purported trapdoor $\tau'_{A_i}$ can be used to recover all the preimages to $y_i$ under $\eta_{\kappa_i}$, with $\kappa_i = (A_i, \kappa_{2,i})$, it is sufficient to show that a procedure exists to efficiently determine whether or not $\tau'_{A_i}$ truly `inverts $A_i$', i.e. recovers $(s, e)$ correctly from all possible $r = A_is + e$ with $e$ having sufficiently small magnitude. This procedure exists in the form of \textsf{Invert} from \cite{trapdoors}.

\section{Completeness and soundness of Protocol \ref{protocol:mf16}}
\label{appendix:chernoff-bound}

For notational convenience, define $\alpha = \frac{a}{\sum_s 2 |d_s|}$ and $\beta = \frac{b}{\sum_s 2 |d_s|}$. Fix an arbitrary state $\rho$ sent by the prover. For $j=1,\ldots,m$ let $X_j$ be a Bernoulli random variable that is $1$ if the $j$-th measurement from step 4 of Protocol \ref{protocol:mf16} yields $-\mathrm{sign}(d_j)$ and $0$ otherwise. Let $X = \sum_{j=1}^m X_j$ and $B_j = \textrm{E}[X|X_j,\ldots,X_1]$. Then $(B_1,\ldots,B_m)$ is a martingale. Applying Azuma's inequality, for any $t\geq 0$
\begin{align*}
\Pr\big( |X - \textrm{E}[X]| \geq  t \: \big) &\leq e^{-\frac{t^2}{2m}}\;.
\end{align*}
In the case of an instance $x\notin L$, as mentioned in the main text $\textrm{E}[X_j] \leq \frac{1}{2}-\beta$. Choosing $t = \frac{1}{2}m(\beta - \alpha)$, it follows that in this case
\begin{align*}
\Pr\Big( X \leq \frac{1}{2}m(1 - \beta - \alpha) \Big) &\leq 2e^{-m(\beta-\alpha)^2/8}\;.
\end{align*}

%

Since $\beta - \alpha$ is inverse polynomial, by \cite{proof-of-gap}, the right-hand side can be made exponentially small by choosing $m$ to be a sufficiently large constant times $\frac{|x|}{(\beta - \alpha)^2}$. The soundness of Protocol \ref{protocol:mf16} follows.

Completeness follows immediately from a similar computation using the Chernoff bound, since in this case we can assume that the witness provided by the prover is in tensor product form. 

\section{Commitment scheme}
\label{appendix:commitment}
We provide an informal description of a generic form for a particular (and commonly seen) kind of commitment scheme. The protocol for making a commitment under this scheme requires three messages in total between the party making the commitment, whom we refer to as the \emph{committer}, and the party receiving the commitment, whom we call the \emph{recipient}. The first message is an initial message $i$ from the recipient to the committer; the second is the commitment which the committer sends to the recipient; and the third message is a reveal message from the committer to the recipient. The scheme consists of a tuple of algorithms $(\mathsf{gen, initiate, commit, reveal, verify})$ defined as follows:
\begin{itemize}[itemsep=1pt, topsep=0pt]
\item $\mathsf{gen(1^\ell)}$ takes as input a security parameter, and generates a public key $pk$.
\item $\textsf{initiate}(pk)$ takes as input a public key and generates an initial message $i$ (which the recipient should send to the committer).
\item $\textsf{commit}(pk, i, m, s)$ takes as input a public key $pk$, an initial message $i$, a message $m$ to which to commit, and a random string $s$, and produces a commitment string $z$.
\item $\textsf{reveal}(pk, i, z, m, s)$ outputs the inputs it is given.
\item $\textsf{verify}(pk, i, z, m, s)$ takes as argument an initial message $i$, along with a purported public key, commitment string, committed message and random string, evaluates $\textsf{commit}(pk, i, m, s)$, and outputs 1 if and only if $z = \textsf{commit}(pk, i, m, s)$.
\end{itemize}

For brevity, we sometimes omit the public key $pk$ and the initial message $i$ as arguments in the body of the paper. The commitment schemes which we assume to exist in the paper have the following security properties:

\begin{itemize}[itemsep=1pt, topsep=0pt]
\item \emph{Perfectly binding}: If $\textsf{commit}(pk, i, m, s) = \textsf{commit}(pk, i, m', s')$, then $(m, s) = (m', s')$.
\item \emph{(Quantum) computationally concealing}: For any public key $pk \leftarrow \textsf{gen}(1^\ell)$, fixed initial message $i$, and any two messages $m, m'$, the distributions over $s$ of $\textsf{commit}(pk, i, m, s)$ and $\textsf{commit}(pk, i, m', s)$ are quantum computationally indistinguishable.
\end{itemize}

It is known that a commitment scheme with the above form and security properties exists assuming the quantum hardness of LWE: see Section 2.4.2 of \cite{coladangelo2019non}. The commitment scheme outlined in that work is analysed in the common reference string (CRS) model, but the analysis can easily be adapted to the standard model when an initial message $i$ is allowed to pass from the recipient to the committer.

\newcommand{\etalchar}[1]{$^{#1}$}

\end{document}